\documentclass{amsart}

\usepackage{amsmath, amsthm, amssymb, delarray, todonotes, graphicx, subcaption, float, url, hyperref} %load extra symbols and environments
\usepackage[margin=1in]{geometry} %set margins
\usepackage{enumerate}
\usepackage{tikz} % used for triangle
\usetikzlibrary{patterns}
\usepackage{comment}

\usepackage[all]{xy}          % for xy-pic pictures
\xyoption{dvips}

\makeatletter
\renewcommand\l@subsection{\@tocline{2}{0pt}{2pc}{5pc}{}}
\makeatother

\newtheorem*{rep@theorem}{\rep@title}
\newcommand{\newreptheorem}[2]{%
	\newenvironment{rep#1}[1]{%
		\def\rep@title{#2 \ref{##1}}%
		\begin{rep@theorem}}%
		{\end{rep@theorem}}} 

\theoremstyle{plain}
\newtheorem{thm}{Theorem}[section]
\newreptheorem{thm}{Theorem}

\newreptheorem{thmp}{Theorem} 
\newtheorem{prop}[thm]{Proposition}

\newtheorem{cor}[thm]{Corollary}

\theoremstyle{definition}
\newtheorem{defin}[thm]{Definition}
\newtheorem{example}[thm]{Example}
\newtheorem{def/ex}[thm]{Definition/Example}

\theoremstyle{remark}
\newtheorem{rem}[thm]{Remark}
\newtheorem{rems}[thm]{Remarks}

\newcommand{\refS}[1]{Section~\ref{S:#1}}
\newcommand{\refT}[1]{Theorem~\ref{T:#1}}
\newcommand{\refC}[1]{Corollary~\ref{C:#1}}
\newcommand{\refP}[1]{Proposition~\ref{P:#1}}
\newcommand{\refD}[1]{Definition~\ref{D:#1}}

\newcommand{\refEx}[1]{Example~\ref{Ex:#1}}

\newcommand{\R}{{\mathbb R}}
\newcommand{\Z}{{\mathbb Z}}

\newcommand{\calC}{{\mathcal{C}}}

\newcommand{\Ho}{{\operatorname{H}}}
\newcommand{\Ch}{{\operatorname{C}}}

\newcommand{\st}{{\operatorname{st}}}
\newcommand{\lk}{{\operatorname{lk}}}

\newcommand{\im}{{\operatorname{im}}}

\newcommand{\esc}{{\searrow\!\!\!\searrow}}

\begin{document}

%%%%%%%%%%%%%%%%%%%%%%%%%%%%%%%%%%%%%%%%%%%%%%%%%%%%%%%%%%%%%%%%%%%%%%%%%%%%%%%%%%%%%%%%

\title[Politics and topology]{Political structures and the topology of simplicial complexes}

%%%%%%%%%%%%%%%%%%%%%%%%%%%%%%%%%%%%%%%%%%%%%%%%%%%%%%%%%%%%%%%%%%%%%%%%%%%%%%%%%%%%%%%%

\author{Andrea Mock}
\address{Department of Mathematics, Wellesley College, 106 Central Street, Wellesley, MA 02481}
\email{am10@wellesley.edu}
%\urladdr{}

\author{Ismar Voli\'c}
\address{Department of Mathematics, Wellesley College, 106 Central Street, Wellesley, MA 02481}
\email{ivolic@wellesley.edu}
\urladdr{ivolic.wellesley.edu}

%%%%%%%%%%%%%%%%%%%%%%%%%%%%%%%%%%%%%%%%%%%%%%%%%%%%%%%%%%%%%%%%%%%%%%%%%%%%%%%%%%%%%%%%

%%%%%%%%%%%%%%%%%%%%%%%%%%%%%%%%%%%%%%%%%%%%%%%%%%%%%%%%%%%%%%%%%%%%%%%%%%%%%%%%%%%%%%%%%%%%%%%%%%%%%%%%

\begin{abstract}
We use the topology of simplicial complexes to model political structures following \cite{AK:Conflict}. Simplicial complexes are a natural tool to encode interactions in political structures since a simplex can be used to represent a subset of compatible agents. We translate   the wedge, cone, and suspension operations into the language of political structures and show how these constructions correspond to merging structures and introducing mediators. We introduce the notions of the viability of an agent and the stability of a political system and examine their interplay with the simplicial complex topology, casting their interactions in category-theoretic language whenever possible. We also introduce a refinement of the model by assigning weights to simplices corresponding to the number of issues the agents agree on. In addition, homology of simplicial complexes is used to detect non-viabilities, certain cycles of incompatible agents, and the (non)presence of mediators. Finally, we extend some results from \cite{AK:Conflict}, bringing viability and stability into the language of friendly delegations and using homology to examine the existence of $R$-compromises and $D$-compromises.
\end{abstract}

%%%%%%%%%%%%%%%%%%%%%%%%%%%%%%%%%%%%%%%%%%%%%%%%%%%%%%%%%%%%%%%%%%%%%%%%%%%%%%%%%%%%%%%%%%%%%%%%%%%%%%%%

\maketitle 

\tableofcontents

%\nocite{*} % this command forces all references in template.bib to be printed in the bibliography

\parskip=6pt
\parindent=0cm

%%%%%%%%%%%%%%%%%%%%%%%%%%%%%%%%%%%%%%%%%%%%%%%%%%%%%%%%%%%%%%%%%%%%%%%%%%%%%%%%%%%%%%%%

\section{Introduction}\label{S:Intro}

%%%%%%%%%%%%%%%%%%%%%%%%%%%%%%%%%%%%%%%%%%%%%%%%%%%%%%%%%%%%%%%%%%%%%%%%%%%%%%%%%%%%%%%%

Simplicial complexes are an important and versatile tool in various branches of mathematics. The simple definition -- a \emph{simplicial complex} is a collection of nonempty subsets of a finite set containing all the singletons (vertices) and all subsets of sets already in the collection (simplices) -- lends itself to varied analysis and many applications. Since each subset of a simplicial complex can be represented via geometric realization as a topological simplex,  combinatorial and algebraic topology are natural tools to bring to bear when studying these objects.

Simplicial complexes can potentially serve as a model for any situation where objects or entities interact in some way. One can think of this  as the higher-dimensional analog of graphs which successfully model networks of pairwise interactions (edges) among objects (vertices).  However, a graph cannot capture the situation where larger subsets of objects interact, while a simplicial complex is ideally suited for this since each of its simplices represents such a multi-fold interaction.

The use of simplicial complexes in modeling complex interactions has been booming in recent years. Applications in fields as disparate as topological data analysis \cite{Carlsson:TDA, Carlsson:TDAHomotopy, Munch:TDA, OPTGH:TDAOverview}, signal processing \cite{BS:SimplicialSignals, BSC:SimplicialSignals,  GGB:NeuroSimplex, JKP:SimplicialSignals, MDBRS:Collaboration}, and neuroscience \cite{GGB:NeuroSimplex} abound. Social science applications are also starting to emerge; connections to game theory are well-established \cite{Egan:NashComplexes, FHN:GameComplex, Martino:GamesComplexes}, and recent work uses simplicial complexes to model social communication and opinion dynamics  \cite{HG:Discourse, IPBL:SimplicialContagion, WZLS:SocialSimplicial}.

For the purposes of our work, the most relevant application of simplicial complexes is the paper \cite{AK:Conflict} by Abdou and Keiding. The authors in this article consider a set of agents in a political system and compatibilities among them. If a subset of agents is compatible, they are able to coexist and carry out the processes in their mandate, such as negotiation or passage of legislation. A subset of compatible agents is called a \emph{viable configuration}. The system of viable configurations can then readily be modeled by a simplicial complex. 

The focus of \cite{AK:Conflict} is the \emph{strong collapse} operation, a process of eliminating vertices that are dominated by others in a suitable sense. Strong collapses are a fairly recent notion in the theory of simplicial complexes, due to Barmak and Minian \cite{BM:StrongHtopy} (see also \cite{B:FiniteTop}). These collapses are stricter than the classical collapse construction of J.H.C.~Whitehead \cite{W:Simplicial}. Their advantage is that they characterize simplicial maps that are \emph{contiguous}, and such maps are, in turn, the correct analog of homotopic maps in the category of simplicial complexes. Strong collapses have been used, for example, in topological data analysis \cite{BP:CollapsePersistence} and the study of the Lusternik-Schnirelmann category of simplicial complexes \cite{FMV:L-SSC}.

Abdou and Keiding use strong collapses to develop the idea of a \emph{friendly delegation} where one agent gives up their standing in favor of another, more centrally placed agent. This leads to the notion of a \emph{represented compromise} and, if friendly delegations occur iteratively, the notion of a \emph{delegated compromise}. The authors use various topological tools to examine the conditions under which these compromises may or may not occur, as well as their implications.

The goal of the present work is to further exploit the topology of simplicial complexes as a model for political structures. We interpret some standard topological constructions (simplicial map, wedge, join, cone, suspension) in this context and study their consequences on political systems. For example, the wedge corresponds to merging political structures via representative agents and the cone models the introduction of a mediator. 

We also define the notion of the viability of an agent via the star of a vertex and the stability of a structure via the $f$-vector of a complex. Much of what we do is devoted to interpreting the interaction of these concepts with the basic constructions on simplicial complexes. For example, we show that introducing mediators or merging structures increases the stability.

One key feature in this work is that we bring homology into the picture. Homology is one of the most useful algebraic invariants of topological spaces. When restricted to simplicial complexes, it lends itself to a relatively easy combinatorial definition and analysis that uses only linear algebra. Since homology is homotopy invariant, it is not quite powerful enough to detect strong equivalences, namely equivalences that can be realized as sequences of strong collapses. However, the presence of nontrivial homology means that a simplicial complex is not contractible, and hence not strongly contractible, and this turns out to convey useful information about the incompabilities in the structure. 

In particular, the presence of nontrivial chains representing homology classes can be directly interpreted as the existence of subsets of agents that are non-viable in specific ways. Much of these observations can also be stated in terms of the \emph{core} of the simplicial complex, which is the smallest subcomplex in which it is no longer possible to perform strong collapses. Homology is essentially able to tell what the core looks like in terms of the number and size of sets of incompatible agents. One can use this for strategic placement of mediators for the greatest impact; this is analogous to the basic algebraic topology procedure of coning off nontrivial cycles as a way to eliminate homology.

Homology is also useful for detecting when certain compromises within structures are not possible. The notions of compromises were defined in \cite{AK:Conflict}, and applying homology to them is one of the ways we revisit and elaborate on the original work by Abdou and Keiding that motivates this paper. We  also examine how our definitions of viability and stability mesh with their notion of friendly delegations.

We provide many examples throughout, as well as commentary on when the parallels between the worlds of simplicial complexes and political stuctures appear to be successful and when they seem to only go so far (see, for example, Remarks \ref{R:DeficientDefiniton} and \ref{R:ViabilityStabilityDeficient}). In particular, we frequently offer potential directions of improvement to the definitions of viability and stability. 

One promising avenue is to consider \emph{weighted} simplicial complexes in which the weights record compatibilities in a more refined way. With more than one issue present, agents may agree on some of them and not the others, and the weights keep track of this. With this in mind, we define stronger notions of  viability and stability and hope to continue to examine their merits in future work. 

Since we consider this paper as only one of the initial steps (along with \cite{AK:Conflict}) in what is potentially a rich area of study, throughout the paper we discuss many directions of investigation that could bring sophisticated topology into the realm of political science in an unprecendented way. The dictionary between simplicial complexes and political structures appears to be of great import, and this paper should be thought of as an invitation for its continued exploration.

%%%%%%%%%%%%%%%%%%%%%%%%%%%%%%%%%%%%%%%%%%%%%%%%%%%%%%%%%%%%%%%%%%%%%%%%%%%%%%%%%%%%%%%%

\subsection{Organization of the paper}

%%%%%%%%%%%%%%%%%%%%%%%%%%%%%%%%%%%%%%%%%%%%%%%%%%%%%%%%%%%%%%%%%%%%%%%%%%%%%%%%%%%%%%%%

We assume the reader is familiar with the basic notions of (algebraic) topology such as topological spaces, maps, homeomorphisms, homotopic maps, homotopy equivalences, categories, functors, etc.  One central idea we will import from algebraic topology is that of homology, and we will provide a brief overview of its construction.  For more on basic point-set and algebraic topology, see \cite{Hatcher, Munkres:Topology}.

The paper is organized as follows:
\begin{itemize}
\setlength\itemsep{4pt}
\item \refS{SimplicialComplexes} provides a brief review of the basics of the topology of simplicial complexes. A reader who is familiar with this material may skip this section, and a reader who is not can find more details in Appendix \ref{A:SimplicialComplexes}.

\item The translation of the language of simpicial complexes to that of political structures is the subject of \refS{SimplicialModel}. In \refS{Dictionary}, we set up the basic dictionary (mostly following \cite{AK:Conflict}). In \refS{Stability}, we define the viability of an agent and the stability of a political structure. Both these notions essentially count the number of simplices; the advantages and drawbacks of such simple definitions (along with suggestions for possible improvements) are provided throughout the paper. In \refS{PoliticsMods}, we study the interplay of viability and stability with operations on political structures, mainly how they interact with merging and introducing mediators into structures (Propositions \ref{P:WedgeStability}, \ref{P:ConeStability}, and \ref{P:SubstructureConeStability}).
\item The improvement of the model to the setting of weighted political structures is the subject of \refS{SimplicialModelIssues}. This is in many ways a better model because, unlike what appears in the paper so far, it takes into account the existence of multiple issues on which agent may or may not agree. We define weighted viability and stability in (\refD{WeightedStability}). Many possible promising directions of investigation involving signals, sheaves, generalized power indices, agreement distributions, and Markov chains are outlined at the end of the section.
\item \refS{StructuresHomology} uses homology to gain insight into political structures. We observe that non-zero Betti numbers detect certain patterns in the non-viability of agents (\refT{StructuresHomology}) as well as the absence of mediators  (\refP{NoMediatorHomology}).
\item \refS{Delegations} extends some results from \cite{AK:Conflict}. After reviewing the basic terminology of friendly delegations, we draw on earlier results to examine how this notion interacts with that of viability and stability (\refT{DelegationStability}). Further interplay, this time with mediators and mergers interacting with $D$-compromises, is also given (\refP{WedgeMediatorCompromise}). Finally, homology is again brought into the picture to show how non-zero Betti numbers prevent the existence of $R$-compromises and $D$-compromises. Various comments and questions about how the definitions and results from \cite{AK:Conflict} might relate to ours are also provided. 
\item Some potential future directions of investigation, beyond those discussed in \refS{SimplicialModelIssues}, are listed in \refS{Future}. These include bringing more theory of simplicial complexes into the picture, including the poset structure, barycentric subdivisions, and the Stanley-Reisner ring. We also raise the question of what effect relaxing the construction of strong collapses  to the more standard one of ordinary collapses might have on our models and results.

\item Appendix \ref{A:SimplicialComplexes}, reviewing the topology of simplicial complexes, is included for completeness and self-containment. Section \ref{S:SimplCplx} contains basic definitions and examples, such as the star, link, and deletion. Geometric realizations are covered in Section \ref{S:Realization}, while \refS{ComplexOperations} contains an overview of the important and familiar  constructions of join, wedge, cone, and  suspension (their analogs in the category of political structures can be found in in \refS{PoliticsMods}). Section \ref{S:SimplMap} is dedicated to simplicial maps. The important notions here are those of contiguity and strong equivalence (\refD{StrongEquivalence}), as well as the result, due to Barmak \cite{B:FiniteTop}, that strong equivalences can be realized as sequences of elementary strong collapses (\refT{StrongCores}). Elementary collapses will be used in \refS{Delegations} in the context of friendly delegations and compromises. The appendix also includes 
some independent results about cores and collapsibility of simplicial complexes (\refP{PushoutCollapse} and \refC{WedgeCollapse}). Finally, in Section \ref{S:Homology}, we review the construction of homology and recall how it can be calculated using simple linear algebra.
\end{itemize}

%%%%%%%%%%%%%%%%%%%%%%%%%%%%%%%%%%%%%%%%%%%%%%%%%%%%%%%%%%%%%%%%%%%%%%%%%%%%%%%%%%%%%%%%

\subsection{Acknowledgments}
The authors are grateful to Franjo \v Sar\v cevi\'c for pointing out several omissions and errors in an earlier version of this paper, as well as to the reviewers for helpful comments and suggestions. 
The second author would like to thank the Simons Foundation for its support.

%%%%%%%%%%%%%%%%%%%%%%%%%%%%%%%%%%%%%%%%%%%%%%%%%%%%%%%%%%%%%%%%%%%%%%%%%%%%%%%%%%%%%%%%

\section{Simplicial complexes}\label{S:SimplicialComplexes}

%%%%%%%%%%%%%%%%%%%%%%%%%%%%%%%%%%%%%%%%%%%%%%%%%%%%%%%%%%%%%%%%%%%%%%%%%%%%%%%%%%%%%%%%

In this section, we define simplicial complexes and briefly review some of theory behind them.  These combinatorially defined objects are  useful since they model most topological spaces one cares about. We will try to avoid bringing the full power and generality of topology into the picture and will ground ourselves in combinatorics as much as possible. 

The material reviewed here will will be familiar to a topologist. For a reader who is not as acquainted with these ideas, the brevity of the exposition in this section is supplemented by the content in Appendix \ref{A:SimplicialComplexes}, where more detail and examples are given. The reason for relegating most of the background on simplicial complexes to the appendix is that we wanted to get to the applications to political systems as soon as possible.

%%%%%%%%%%%%%%%%%%%%%%%%%%%%%%%%%%%%%%%%%%%%%%%%%%%%%%%%%%%%%%%%%%%%%%%%%%%%%%%%%%%%%%%%

%\subsection{Basic definitions}\label{S:SimplCplx}

%%%%%%%%%%%%%%%%%%%%%%%%%%%%%%%%%%%%%%%%%%%%%%%%%%%%%%%%%%%%%%%%%%%%%%%%%%%%%%%%%%%%%%%%

\begin{defin}\label{D:SimplicialComplex}
An \textbf{(abstract) simplicial complex} $K = (V, \Delta)$ consists of a 
finite set $V$ whose elements are called \emph{vertices}  and a set $\Delta$ of subsets of $V$ called \emph{simplices}  satisfying
\begin{enumerate}
\item Elements of $V$ are in $\Delta$;
\item If $\sigma\in\Delta$ and $\tau\subset\sigma$, then $\tau\in\Delta$.
\end{enumerate}
\end{defin}

If $|V|=k+1$, we will often label the elements of $V$ by $v_0, ..., v_{k}$.  The choice of enumeration of the elements will not be relevant.
Since elements of $V$ are already listed in $\Delta$, we will simply identify $K$ with $\Delta$ and will not write down $V$ explicitly.

%\begin{example}\label{Ex:AbsSimplCplx}
%An example of a simplicial complex on the vertex set  $V=\{v_0, v_1, ..., v_7\}$ is 
%\begin{align*}
%K= & \{
%\{v_0, v_1, v_2, v_3\},
%\{v_0, v_1, v_2\}, 
%\{v_0, v_1, v_3\}, 
%\{v_0, v_2, v_3\}, 
%\{v_1, v_2, v_3\},
%\{v_0, v_1\}, 
%\{v_0, v_2\}, 
%\{v_0, v_3\}, 
%\{v_1, v_2\}, \\ & 
%\{v_1, v_3\}, 
% \{v_2, v_3\}, 
%\{v_3, v_4\}, 
%\{v_3, v_5\}, 
%\{v_4, v_5, v_6\}, 
%\{v_4, v_5\}, 
%\{v_4, v_6\}, 
%\{v_5, v_6\}, 
%\{v_4, v_7\},  \\
%& \{v_0\}, \{v_1\}, \{v_2\}, \{v_3\}, \{v_4\}, \{v_5\}, \{v_6\}, \{v_7\}
%\}
%\end{align*}
%
%
%A non-example is $V=\{v_0, v_1, v_2, v_3, v_4\}$ and a collection of subsets
%$$
%K=\{ 
%\{v_1, v_2, v_3, v_4\},
%\{v_1, v_3, v_4\},
%\{v_0, v_2\},
%\{v_1, v_2\},
%\{v_1, v_4\},
%\{v_2, v_3\},
%\{v_3, v_4\},
%\{v_0\}, \{v_1\}, \{v_2\}, \{v_3\}, \{v_4\}
%\}
%$$
%The reason this is not a simplicial complex is that, for example, even though $\{v_1, v_2, v_3, v_4\}$ is in $K$, its subset $\{v_1,v_2, v_4\}$ is not.
%
%\refEx{TopSimplCplx} shows a topological visualization of these examples.
%\end{example}
%
%

%\begin{example}\label{Ex:Standard}
One basic example is that of the \emph{(standard) $n$-simplex}, where $|V|=n+1$ and $K=\mathcal P_0(V)$, the power set of $V$ without the empty set.  
%Thus $K$ contains all possible nonempty subsets of $V$. 
%\end{example}

%
%\begin{example}
%Recall that a \emph{graph} $\Gamma$ consists of a set $V$ of vertices and a set $E$ of edges defined as 
%$$
%E\subset\{\{v_1,v_2\}\colon v_1,v_2\in V, v_1\neq v_2\}.
%$$ 
%It is immediate from the definitions that a graph is precisely a 1-complex.
%\end{example}
%
%

Some standard terminology, including face, star, $n$-skeleton, etc.~is recalled in \refD{ComplexStuff}. One definition we will go back to frequently is that a vertex $v$ is \emph{dominated by a vertex $w$} if every maximal simplex that contains $v$ also contains $w$.

To a simplicial complex $K$ one can associate a topological space $|K|$ called the \emph{geometric realization} (or the \emph{polyhedron})  of $K$.  To construct the realization, one takes a convex hull of $n+1$ affinely independent points in a Euclidean space for each $n$-simplex in $K$. The details of the construction, as well as examples, are provided in Appendix \ref{S:Realization}. We will not make a distinction between the abstract simplicial complex $K$ and its realization $|K|$ and will use $K$ for both.

There are several useful constructions one can perform on simplicial complexes, such as the \emph{wedge} and the \emph{cone}. For simplicial complexes $K$ and $L$, the wedge $K\vee L$ is the union of $K$ and $L$ with two vertices, one from each complex, identified. The cone $CK$ on $K$ is obtained from $K$ by joining all the simplices of $K$ to a new disjoint vertex. The details and examples of these constructions are recalled in Appendix \ref{S:ComplexOperations}. They will correspond to operations on political structures in \refS{PoliticsMods}.

Simplicial maps, recalled in Appendix \ref{S:SimplMap}, are a way to compare simplicial complexes, just as continuous maps are a way to compare topological spaces. One class of such maps is of special interest to us:  Simplicial maps $\phi$ and $\psi$ from $K$ to $L$ are \emph{contiguous} if $\phi(\sigma)\cup \psi(\sigma)$ is a simplex in $L$ whenever $\sigma$ is a simplex in $K$. 

Contiguity leads to the notion of \emph{strong equivalence}, which is in turn related to the notion of a \emph{strong collapse} -- a sequence of removals of dominated vertices. All of these definitions and contructions are reviewed in Appendix \ref{S:Collapse} and will be relevant in \refS{Delegations}.

One of the most central concepts in algebraic topology and one of the most important topological invariants of spaces is that of homology. Intuitively, $n$th homology group keeps track of $n$-dimensional holes in a space. As we will see, homology can detect various features in a political structure. A review of homology is supplied in Appendix \ref{S:Homology}.

\section{Modeling political structures with simplicial complexes}\label{S:SimplicialModel}

%%%%%%%%%%%%%%%%%%%%%%%%%%%%%%%%%%%%%%%%%%%%%%%%%%%%%%%%%%%%%%%%%%%%%%%%%%%%%%%%%%%%%%%%%%%%%%%%%%%%%%%%%%%%%%%%%

In this section, we translate the language of simplicial complexes into that of political structures and find analogues of many results and constructions from the earlier sections. The main ideas here are those of viability of agents and  stability of political structures (\refD{Stability}) and how those notions interact with the topology of simplicial complexes (Propositions \ref{P:WedgeStability},  \ref{P:ConeStability}, \ref{P:SubstructureConeStability}).  In \refS{SimplicialModelIssues}, we lay down the beginnings of a more refined point of view that takes into accounts the agents' positions  on various issues. This weighted version of the political structure model appears to have much potential, as outlined at the end of the section. Finally, \refS{StructuresHomology} explores the interaction of homology with political structures. In particular, we show how homology detects non-viable subsets of agents (\refT{StructuresHomology}).

%%%%%%%%%%%%%%%%%%%%%%%%%%%%%%%%%%%%%%%%%%%%%%%%%%%%%%%%%%%%%%%%%%%%%%%%%%%%%%%%%%%%%%%%%%%%%%%%%%%%%%%%%%%%%%%%%

\subsection{Basic dictionary}\label{S:Dictionary}

%%%%%%%%%%%%%%%%%%%%%%%%%%%%%%%%%%%%%%%%%%%%%%%%%%%%%%%%%%%%%%%%%%%%%%%%%%%%%%%%%%%%%%%%%%%%%%%%%%%%%%%%%%%%%%%%%

Here we set up the language of political structures and show how simplicial complexes can model it. The terminology we use is directly from \cite{AK:Conflict}.

Suppose we have a collection of \emph{agents}, which can be thought of as voters, political parties, or members of organizations and structures like legislative bodies or boards of directors. We wish to model the question of compatibilities of these agents, their potential for coexistence. If agents can coexist, then the work of the system can continue. Within the system, agents may or may not agree on particular issues or sets of issues, but the potential for compromise and formation of coalitions exists. If the agents are incompatible, then one has a political stalemate, a polarization that impedes the functioning of the system. The question of compatibility can arise after a period of volatility, like a war or some other dramatic upheaval, or even after an election.

In between the most desirable situation where all agents are compatible and the least desirable one where no subset of them is compatible are the possibilities where agents fall into various subsets of compatibilities. The more of these subsets there are and the greater the number of agents in them means that the system is closer to being fully functional.

To set up the framework, we will make the assumptions that participation of all agents in the political process is necessary (for example, everyone has to cast a vote) and that, if some agents are compatible, then any fewer of those agents are also compatible.  

We can then make the following definition after  \cite[Definition 1]{AK:Conflict}.

\begin{defin}\label{D:PoliticalStructure}
A \emph{political structure} $P=(A,\calC)$ is a finite set of agents $A$ along with a collection $\calC$ of subsets of those agents, called \emph{viable configurations} satisfying:
\begin{enumerate}
\item Each agent is a viable configuration on their own;
\item If some agents form a viable configuration, so does any subset of them. 
\end{enumerate}
\end{defin}

Comparing this to \refD{SimplicialComplex}, it is clear that there is a correspondence between political structures $(A,\calC)$ and simplicial complexes $(V,\Delta)$ given by
\begin{align*}
\text{agents $a_i\in A$} & \longleftrightarrow  \text{vertices $v_i\in V$}\\
\text{viable configurations $\gamma\in\calC$}  & \longleftrightarrow  \text{simplices $\sigma\in \Delta$}
\end{align*}

\begin{rem}\label{R:DeficientDefiniton}
One direction in which further work might improve our setup is in amending \refD{PoliticalStructure} to remove the assumption that, if a collection of agents is compatible, so is any subset of them. Indeed, an agent might decide to join a coalition based on who is already in it; for example, agent $a_1$ might decide to join the coalition $\{a_2, a_3\}$ but not be compatible with $a_2$ or $a_3$ individually. A viable configuration might also have a core subset that holds it together in some way and the removal of that subset might cause the configuration to dissolve. In addition, a smaller subset of agents may not be able to arrive at a concensus, and hence be incompatible, and needs the larger group to sway the vote in some direction. These complications indicate that, instead of using simplicial complexes, one might try to develop the theory of political structures in the more general setting of, say, hypergraphs.  Related considerations motivate our development of the weighted simplicial complex model for political structures; see the opening paragraph to Section \ref{S:SimplicialModelIssues}.
\end{rem}

\begin{defin}\label{D:PoliticalMap}
Suppose $P$ and $Q$ are political structures. A \emph{political structure map}  is a function $p\colon P\to Q$ that sends agents to agents and viable configurations to viable configurations.
\end{defin}

In analogy with \refD{SimplMap} and the discussion following it, this definition says that if $\{a_{i_1}, ...,a_{i_n}\}$ is a viable configuration in $P$, then  $\{p(a_{i_1}), ...,p(a_{i_n})\}$ must be a viable configuration in $Q$.

A political structure map is designed to compare political structures.  Agents from one political system are mapped to the corresponding, or like-minded, agents in another system. The correspondence is further captured by the requirement that, if certain agents are compatible in one structure, so are the corresponding ones in the other. In addition to enabling comparisons between existing systems, political structure maps might model the before and the after of some political change or turmoil which consolidates agents or creates new ones but preserves alliances between those agents that survive.

A political structure map need not be injective or surjective.  If it is not injective, that means that two or more agents correspond to a single agent, i.e.~they are in some way consolidated or have merged in the structure following a political change. If a map is not surjective, that means that new agents emerged from the process.  An isomorphism is a map that identifies the structures as identical, with a one-to-one correspondence between agents and viable configurations. An inclusion map means that a structure is absorbed into a larger one, but with its system of viabilities still intact.

It is easy to see that the collection of political structures along with political structure maps forms a category which we will denote by  $\operatorname{PolStr}$.  Since political structures correspond to simplicial complexes and political structure maps to simplicial maps, we in fact immediately have an equivalence of categories
$$
\operatorname{PolStr}\simeq \operatorname{SCom}.
$$
We can thus carry over all the constructions from the category of simplicial complexes into the realm of political structures, including the product, pushout, join, as well as functors like the geometric realization and  homology. 

\begin{example}\label{Ex:PoliticalStructures}
Figure \ref{F:PoliticalStructures} depicts two political structures. In the left one, agent $a_0$ can coexist with agent $a_1$, but is in conflict with all other agents, who are all compatible among themselves.  The right picture is similar, except agents $a_1$, $a_2$, and $a_3$ are compatible in pairs, but not all together.  Each of those agents is thus open to pairwise coalitions, but not to the third agent joining in.

\begin{figure}[h]
\centering
\includegraphics[width=1.7in]{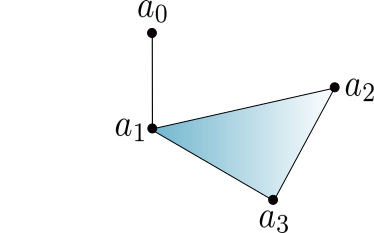}
\ \ \ \ \ \ \ 
\includegraphics[width=1.7in]{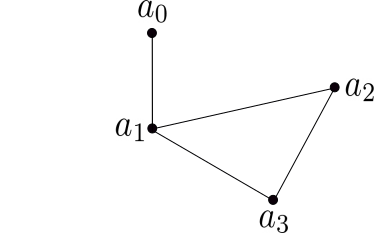}
\caption{Examples of political structures.}
\label{F:PoliticalStructures}
\end{figure}

\end{example}

\begin{example}\label{Ex:MiddleEast}
Figure \ref{F:MiddleEast} is a depiction of the compatibilites of the various actors in the Middle East arena. The ``friendship chart'' that was used to create this simplicial complex was published in 2014,\footnote{The chart can be found at \url{http://www.slate.com/blogs/the_world_/2014/07/17/the_middle_east_friendship_chart.html}} and it may be  dated by now. Nevertheless, it illustrates the idea of modelling compatibilities with simplicial complexes. The agents are
\begin{center}
\begin{tabular}{lll}
$a_0 = $ Hamas                 & $a_5 = $ Egypt      & $a_9 = $ Syria\\
$a_1 = $ Turkey                & $a_6 = $ Hezbollah  & $a_{10} = $ Israel\\
$a_2 = $ Palestinian Authority & $a_7 = $ Iraq       & $a_{11} = $ Al-Qaida\\
$a_3 = $ Saudi Arabia          & $a_8 = $ USA        & $a_{12} = $ ISIS\\
$a_4 = $ Iran &  & \\
\end{tabular}
\end{center}
We will revisit this example in Examples \ref{Ex:MiddleEastViability}, \ref{Ex:MiddleEastBetti}, and \ref{Ex:MiddleEastCompromise}.

\begin{figure}[h]
\centering
\includegraphics[width=2.5in]{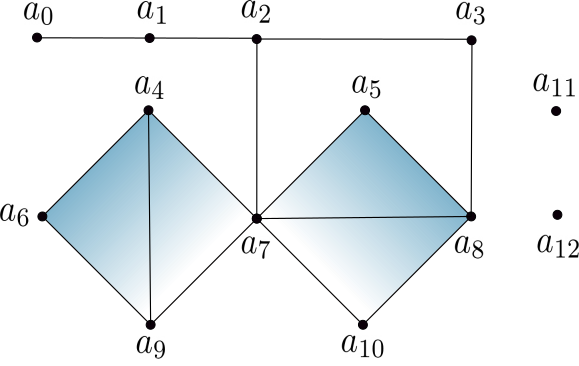}
\caption{A simplicial complex model for the relationships in the Middle East.}
\label{F:MiddleEast}
\end{figure}

\end{example}

\begin{example}\label{Ex:BasicPoliticalStructures} Suppose there are $k+1$ agents.
\begin{itemize}
\item The situation where all agents form a viable configuration corresponds to the standard $k$-simplex. 
\item The situation where no subset of agents is viable corresponds to the 0-dimensional simplicial complex, or a discrete set of $k+1$ vertices.
\item More generally, if there are $n$ subsets of agents and agents in one subset are not compatible with agents in any other subset, then the corresponding simplicial complex has (at least) $n$ disconnected components.
\item The situation where all viable configurations consist of at most two agents is a 1-dimensional simplicial complex, or a graph.

\end{itemize}
\end{example}

In the first situation from the previous example, namely when all agents form a viable configuration, we will say that $P$ is \emph{fully viable}. 

Several notions from \refD{ComplexStuff} also carry over and have interpretations in the setting of political structures.

\begin{defin}\label{D:PoliticalSystemStuff}\ 

\begin{itemize}
\setlength\itemsep{4pt}
\item The \emph{dimension} $\dim P$ of a political structure $P$ is one less than the size of its largest viable configuration. This is the largest collection of agents that is willing to compromise or enter a coalition. Since the greatest possible dimension of a simplicial complex with $k+1$ agents is $k$ ($k$-simplex) and this represents the ideal political situation, the number $k-\dim P$ is one measure of the deviation of $P$ from being a functioning system. We will call this number the \emph{defect of $P$} and denote it by $\operatorname{def}(P)$.
%\item The collection of all viable configurations consisting of $n+1$ agents is $P^{(n)}$, the \emph{$n$-skeleton} of $P$.
\item The $f$-vector of $P$ records the number of viable configurations of $P$ of each size.
\item If there is a viable configuration that is not a subconfiguration of any other viable configuration, then this is a \emph{maximal configuration}. Each maximal configuration represents a subset of agents who are willing to compromise, but no other agent would be willing to join their coalition.  
\item An agent $a_i$ is \emph{dominated} by an agent $a_j$ if all the largest possible coalitions that $a_i$ is willing to join also have $a_j$ in them. This indicates that $a_j$ is overall better positioned in the political system, and we say that $a_j$ is more \emph{central} than $a_i$.
\item The \emph{star} of an agent $a_i$ is the collection of its compatibilities, i.e.~the collection of all the potential coalitions $a_i$ is willing to enter.
\item The \emph{deletion} of an agent $a_i$ is its removal from the political system. The removal of all the simplices that have $a_i$ as a vertex means that this agent can no longer be present in any potential coalitions.
\end{itemize}

\end{defin}

%%%%%%%%%%%%%%%%%%%%%%%%%%%%%%%%%%%%%%%%%%%%%%%%%%%%%%%%%%%%%%%%%%%%%%%%%%%%%%%%%%%%%%%%%%%%%%%%%%%%%%%%%%%%%%%%%

\subsection{Stability of political structures}\label{S:Stability}

%%%%%%%%%%%%%%%%%%%%%%%%%%%%%%%%%%%%%%%%%%%%%%%%%%%%%%%%%%%%%%%%%%%%%%%%%%%%%%%%%%%%%%%%%%%%%%%%%%%%%%%%%%%%%%%%%

As seen in Example \ref{Ex:PoliticalStructures}, the most harmonious situation is when all $k+1$ agents can coexist, and this is modeled by the $k$-simplex where all possible simplices exist between all subsets of agents. At the other end of the spectrum is the situation where there is no agreement at all among the agents, with the discrete set of agents and no simplices connecting them representing it.  Thus the more simplices there are in the complex modeling it, the more stable the system is.  

As a more concrete example, the 2-simplex in Example \ref{Ex:PoliticalStructures} represents a potential coalition of agents $a_1$, $a_2$, and $a_3$, which indicates a more stable system than the one in the right complex where only potential coalitions of two agents exist.

These observations motivate the following definition. 

%Recall the notion of the star $\st(v)$ of a vertex $v$ and the $f$-vector of a simplicial complex from \refD{ComplexStuff}.

\begin{defin}\label{D:Stability}
Given a $d$-dimensional political structure $P=(A,\mathcal C)$ with the set of agents $A=\{a_0, ..., a_k\}$, define the \emph{viability of $a_i$} to be
\begin{equation}\label{E:AgentStability}
\operatorname{via}(a_i)=\frac{1}{2^k-1}(|\operatorname{st}(a_i)|-1).
\end{equation}

Define the \emph{stability of $P$} to be

\begin{equation}\label{E:ComplexStability}
\operatorname{stab}(P)=\frac{1}{2^{k+1}-(k+2)}\left(\sum_{i=0}^{d}f_i-(k+1)\right).
%\operatorname{stab}(P)=\frac{1}{2^{k}(k+1)}\left( \sum_{a_i\in A}|\operatorname{st}(a_i)| -(k+1) \right)
\end{equation}

\end{defin}

The viability of $a_i$ is simply the number of simplices that have $a_i$ as a vertex, but normalized to take values between 0 and 1. Greater value of $\operatorname{via}(a_i)$ indicates that, even if some agents abandon or become incompatible with $a_i$, this agent still has the potential of forming other coalitions due to its large star, i.e.~due to its many existing compatibilities. This is the sense in which we think of it as a ``more viable'' agent than one whose star does not have as many elements. 

The stability of $P$ is the  total number of simplices of $P$, again normalized. This is essentially the sum of the viabilities of all the agents, but taking into account the overcount of simplices. Greater $\operatorname{stab}(P)$ indicates more compatibilities among agents and more willingness to form coalitions.

\begin{example}
When all the agents are compatible, $|\operatorname{st}(a_i)|=2^k$ for all $a_i$, and so the viability of each agent is 1. In this case, the number of simplices is $\sum_{i=0}^{k}f_i=2^{k+1}-1$ (number of simplices of a $k$-simplex, i.e.~the number of nonempty subsets of a set with $k+1$ elements), and so the stability of $P$ is also 1.

When no agents are compatible, $|\operatorname{st}(a_i)|=1$ for all $a_i$, and the viability of each agent is thus 0. Now the number of simplices is $\sum_{i=0}^{k}f_i=k+1$ (there are that many vertices and no other simplices), and so $\operatorname{stab}(P)=0$.
\end{example}

\begin{example}
In the left complex of Example \ref{Ex:PoliticalStructures}, we have $\operatorname{via}(a_0)=1/7$, $\operatorname{via}(a_1)=4/7$, $\operatorname{via}(a_2)=\operatorname{via}(a_3)=3/7$, and $\operatorname{stab}(P)=5/11$.

In the right complex, $\operatorname{via}(a_0)=1/7$, $\operatorname{via}(a_1)=4/7$, $\operatorname{via}(a_2)=\operatorname{via}(a_3)=3/7$, and $\operatorname{stab}(P)=4/11$.
\end{example}

\begin{example}\label{Ex:MiddleEastViability}
In Example \ref{Ex:MiddleEast}, Iraq is the most viable agent, with $\operatorname{via}(a_7)=9/4095.$
\end{example}

\begin{rems}\label{R:ViabilityStabilityDeficient}
Even though the star and the $f$-vector, and hence $\operatorname{via}(a_i)$ and $\operatorname{stab}(P)$, are isomorphism invariants, they are not complete invariants.  For example, if 
$P_1=\{\{a_0, a_1\}, \{a_0, a_2\}, \{a_0, a_3\}, \{a_0\}, \{a_1\},\{a_2\},\{a_3\} \}$ 
and 
$P_2=\{\{a_0, a_1\}, \{a_1, a_2\}, \{a_0, a_2\}, \{a_0\}, \{a_1\},\{a_2\},\{a_3\} \}$,
we have 
$\operatorname{stab}(P_1)=\operatorname{stab}(P_2)$.  However, $P_1$ has a vertex of greatest stability in both systems, namely $a_0$, and no isolated agents (while $a_3$ is isolated in $P_2$). In this sense, $P_1$ can perhaps be regarded as more stable than $P_2$, but our definitions do not take this into account. 

Another observation that can be made is that the removal of $a_0$ from $P_1$ would result in the collapse of the system as $a_1$, $a_2$, and $a_3$ would all become isolated agents. In $P_2$, this removal would still preserve one potential coalition, that between $a_1$ and $a_2$. This indicates that our rudimentary notions of stability could potentially be refined from a basic simplex count to incorporating the effect of a deletion of an agent. We will analyze such deletions in more detail in \refS{Delegations}. 
\end{rems}

%%%%%%%%%%%%%%%%%%%%%%%%%%%%%%%%%%%%%%%%%%%%%%%%%%%%%%%%%%%%%%%%%%%%%%%%%%%%%%%%%%%%%%%%%%%%%%%%%%%%%%%%%%%%%%%%%

\subsection{Merging and mediating political structures}\label{S:PoliticsMods}

%%%%%%%%%%%%%%%%%%%%%%%%%%%%%%%%%%%%%%%%%%%%%%%%%%%%%%%%%%%%%%%%%%%%%%%%%%%%%%%%%%%%%%%%%%%%%%%%%%%%%%%%%%%%%%%%%

In this section, we carry over some notions from \refS{ComplexOperations} into the setting of political structures.

Suppose a political structure $P$ consists of two distinct connected components, $P_1$ and $P_2$, so $P=P_1\amalg P_2$. Suppose an agent $a_0\in P_1$ and an agent $b_0\in P_2$ decide to merge or join forces, effectively becoming a single agent. If we rename the newly formed agent by $a$, we then have a situation that corresponds to the wedge of simplicial complexes $P_1\vee P_2$ (see \refD{Wedge}). The viability of the common agent is greater than the viability of the individual agents that merged prior to the coalition: 
$$
\operatorname{via}(a)=\operatorname{via}(a_0)+\operatorname{via}(b_0)-1.
$$
The stability of the system is also affected; the fact that the two parts of $P$ could not coexist prior to the merger but now can is reflected in the following.

\begin{prop}\label{P:WedgeStability}
Merging agents from disconnected components of a political structure strictly increases stability, i.e.
$$
\operatorname{stab}(P_1\amalg P_2)< \operatorname{stab}(P_1\vee P_2).
$$
\end{prop}

\begin{proof}
Suppose $P_1$ has dimension $d_1$, $k+1$ vertices, and the simplex count vector $(f_0, ..., f_k)$. Suppose $P_2$ has dimension $d_2$, $m+1$ vertices, and vector $(g_0, ..., g_m)$. Then 
$$\operatorname{stab}(P_1\amalg P_2)=\frac{1}{2^{k+m+2}-(k+m+3)}\left(\sum_{i=0}^{d_1}f_i+\sum_{j=0}^{d_2}g_j-(k+m+2)\right).$$ 
The stability of the wedge (which has one fewer vertex) is
$$\operatorname{stab}(P_1\vee P_2)=\frac{1}{2^{k+m+1}-(k+m+2)}\left(\sum_{i=0}^{d_1}f_i+\sum_{j=0}^{d_2}g_j-1-(k+m+1)\right).$$
The expressions in parentheses are the same, and 
$$2^{k+m+2}-(k+m+3)>2^{k+m+1}-(k+m+2)$$
for all $k,m\geq 0$. It follows that 
$$\operatorname{stab}(P_1\amalg P_2)< \operatorname{stab}(P_1\vee P_2).$$
\end{proof}

\refP{WedgeStability} generalizes to the broader case when two political structures merge along two compatible substructures. This corresponds to a pushout of simplicial complexes (\refD{Pushout}). We leave the details of the generalization to the reader. 

An interesting variant of this situation is when $P_1$ and $P_2$ are initially completely separate political structures, but are then forced to merge as a result of some major event such as the unification of Germany in 1990. In this case, the stability of the newly formed political structure necessarily decreases.

To illustrate, the stability of the structure $\{\{a_0,a_1\}, \{a_0\},\{a_1\}\}$ is 1. Take another copy of the same structure, $\{\{b_0,b_1\}, \{b_0\},\{b_1\}\}$ and form the wedge. This gives the structure $\{\{a_1,a\},\{b_1,a\} \{a\},\{a_1\}, \{b_1\}\}$ whose stability is 1/2, less than the stability of the individual structures prior to merging. It is not hard to show that this always happens with the wedge or a pushout of any two political systems that have more than one agent. 

This is reflected in the defect of $P_1\vee P_2$ as well. Suppose $P_1$ and $P_2$ have $k_1+1$ and $k_2+1$ vertices, respectively. If they start out as components of the same system, namely as $P_1\amalg P_2$, then $\operatorname{def}(P_1\vee P_2)=\operatorname{def}(P_1\amalg P_2)$. However, if $P_1\vee P_2$ is the result of merging two independent systems, then 
$$
\operatorname{def}(P_1\vee P_2)=k_1+k_2+1-\max\{\dim(P_1), \dim(P_2)\},
$$
and this number is greater than $k_1-\dim(P_1)$ and $k_2-\dim(P_2)$, the defects of $P_1$ and $P_2$. This indicates that $P_1\vee P_2$ is farther from unanimity on any issue than $P_1$ and $P_2$ were individually.

Next suppose a political structure $P$ has an agent which is compatible with every other agent and which is willing to join any potential coalition. We will call such an agent a \emph{mediator}, since this setup might model a situation where an outside agent is introduced into the system with the role of bringing as many agents together as possible. The resulting structure is precisely $CP$, the cone on $P$. 

For every configuration $\gamma$ in $P$, there is a new configuration $C\gamma$ of one dimension higher in $CP$. But $CP$ has one more vertex than $P$, so this means that
$$
\operatorname{def}(CP)=\operatorname{def}(P).
$$
%However, the viability of each agent in $P$ is doubled since, if $\gamma\in\st_P(a_i)$, then $C\gamma\in\st_{CP}(a_i)$, where the subscript denotes the structure in which the star is being taken. 
The mediator $c$ has the greatest viability in $CP$ since its link $\operatorname{lk}(c)$ is the entire structure $P$. This maximum viability is shared with other vertices if and only if $P$ is itself a simplex, i.e.~it represents the ideal political situation. In this case, $CP$ is a simplex of one higher dimension and it again represents the ideal situation.

\begin{prop}\label{P:ConeStability}
Introducing a mediator increases the stability of a political structure, i.e.
$$
\operatorname{stab}(P)\leq \operatorname{stab}(CP),
$$
with equality holding if and only if $P$ is fully viable.
\end{prop}

\begin{proof} If the number of viable configurations of $P$ is $\sum_{i=0}^{d}f_i$, then the number of viable configurations after the mediator is introduced is $2\sum_{i=0}^{d}f_i+1$. This is because taking the cone doubles the number of simplices (for each simplex in $P$, the cone introduces a simplex of  one higher dimension) and it adds the cone point as an additional new simplex.
We thus have 
$$
%\operatorname{stab}(P)=\frac{1}{2^{k+1}-(k+2)}\left(\sum_{i=0}^{k}f_i-(k+1)\right),\ \ 
\operatorname{stab}(CP)=\frac{1}{2^{k+2}-(k+3)}\left(2\sum_{i=0}^{d}f_i+1-(k+2)\right)
=
\frac{1}{2^{k+2}-(k+3)}\left(2\sum_{i=0}^{d}f_i-(k+1)\right)
$$
so we wish to show
$$
\operatorname{stab}(P)=\frac{\sum_{i=0}^{d}f_i-(k+1)}{2^{k+1}-(k+2)}\   \leq \ \frac{2\sum_{i=0}^{d}f_i-(k+1)}{2^{k+2}-(k+3)}=\operatorname{stab}(CP).
$$
The verification of this inequality is a straightforward exercise that comes down to the fact that $\sum_{i=0}^{d}f_i\leq 2^{k+1}-1$. The equality holds if and only if $P$ is a simplex, i.e.~it is fully viable.
\end{proof}

Similar analysis can be performed when more than one mediating agent is introduced. When there are two such agents, one gets the suspension $\Sigma P$. Since the suspension can be thought of as the pushout of two copies of $CP$ along $P$, we can use \refP{ConeStability} and the fact that the pushout increases stability (see comments following \refP{WedgeStability}) to conclude that the suspension also increases stability. By induction, the same conclusion holds when an arbitrary number of mediating agents is introduced, a situation that is modeled by the join of $P$ with a finite set of vertices or as an iterated pushout of cones on $P$.

Stability can also be used to measure the \emph{impact} of a mediator, by which we simply mean the difference of the stabilities before and after the mediator is introduced.  As expected, the mediator has the greatest impact in the situation when there is no agreement whatsoever, i.e.~the system is the 0-simplex consisting of $k+1$ agents. The stability in this case is 0. After a mediator is brought in, the stability increases to $\frac{k+2}{2^{k+2}-(k+3)}$. 

Contrast this with the situation when the system is almost perfectly stable, with all subsets of agents being viable except the largest one consisting of all the agents. In the case $k=3$, for example, this is modeled by the hollow tetrahedron. The stability of this structure is $\frac{2^{k+1}-(k+3)}{2^{k+1}-(k+2)}$. After a mediator is introduced, the new stability is $\frac{2^{k+2}-(k+5)}{2^{k+2}-(k+3)}$. The increase in stability is much less dramatic than in the previous situation. For example, with four agents, the increase in the case of total disagreement is $5/26\approx 0.19$ and in the case of almost total agreement it is $4/286\approx 0.014$.

Of course, in the situation when the political structure is as harmonious as possible, the mediator has no impact since $C\Delta^k=\Delta^{k+1}$ and the stability is 1 before and after the mediator is introduced.

One could also ask the question of what might happen if a mediator is introduced into a substructure of $P$. This might happen if, for example, certain agents are particularly at odds on an issue and a separate mediation might be required to bring them to the table. This corresponds to taking a cone on a subcomplex of $P$. The situation is less straightforward in this case. 

\begin{prop}\label{P:SubstructureConeStability}
Suppose $P$ is a political structure with $k+1$ agents. If a mediator is introduced into a substructure so that the total number of agents and viable configurations that are left unaffected is least $k+1$, then the stability of $P$ strictly decreases. 
%If that number is $k+1$, then the stability of $P$ strictly decreases.
\end{prop}

\begin{proof}
%We have that certain simplices of $P$ are coned off and that there are at least $k+1$ other simplices.  
Let $F$ be the number of simplices that are coned off and $G$ the number of those that are not. The assumption is therefore that $G\geq k+1$.

The total number of simplices of $P$ is thus $F+G$ and the number of simplices of the new structure (with a mediator in it), call it $P'$, is  $2F+1+ G$. We then want to show
\begin{equation}\label{E:ConeStability}
\operatorname{stab}(P)=\frac{F+G-(k+1)}{2^{k+1}-(k+2)}\   > \ \frac{2F+G-(k+1)}{2^{k+2}-(k+3)}=\operatorname{stab}(P').
\end{equation}
The inequality is equivalent to 
\begin{equation}\label{E:SubstructureConeStability}
(k+1)F>(k+1-G)(2^{k+1}-1).
\end{equation}
If $G\geq k+1$, the right side is negative or zero, while the left side is always positive, so the inequality is true.
%
%If $G= k+1$, then we want to show \eqref{E:ConeStability} holds, but with the inequality flipped. Since $G=k+1$, this simplifies to the inequality $(k+1)F>0$, which is true. This shows $\operatorname{stab}(P)>\operatorname{stab}(P')$ as desired.
\end{proof}

\refP{SubstructureConeStability} is perhaps counterintuitive since introducing a mediator should improve the functioning of a political system, even if the mediator is acting only on a subset of agents.  On the other hand, a mediator might be making the functioning of a subset of agents stronger but, while doing that, weakening the overall structure since it isolates those agents from the rest of the system.

As an example of the edge case $G= k+1$, consider the structure with $k+1=7$ vertices, three of which form a 2-simplex, and with the other four isolated. The stability of this structure is 1/30. Taking the cone on the four isolated vertices means that $G=7$ (the number of simplices in the 2-simplex), which is the same as the number of vertices. The stability of the new structure is 8/247, which is less than 1/30.

The case $G<k+1$ is more subtle. For example, if $P=\{\{a_0, a_1\}, \{a_1, a_2\}, \{a_0\}, \{a_1\}, \{a_2\}  \}$, then taking the cone on the subcomplex $\{\{a_0, a_1\}, \{a_0\}, \{a_1\} \}$ leaves $G=2$ ($<3=k+1$) simplices unaffected. Then $\operatorname{stab}(P)=1/2>5/11=\operatorname{stab}(P')$.

On the other hand,  if $P=\{\{a_0, a_1\}, \{a_0\}, \{a_1\}, \{a_2\}, \{a_3\}  \}$ and the cone is taken on $\{\{a_2\}, \{a_3\}  \}$, then $G=3$ ($<4=k+1$) but the stability increases, namely $\operatorname{stab}(P)=1/11<3/26=\operatorname{stab}(P')$.

%consider the two structures in Figure \ref{SubstructureConeStability} and their cones shows a structure $P$ in the left picture and two structures obtained from $P$ by coning off subcomplexes. In both situations, $G<k+1$, but the stability in the middle picture is greater than that of $P$ while in the last picture it is not.
%
%
%\begin{figure}[h]
%\centering
%\includegraphics[width=2in]{SubstructureConeStability}
%\caption{}
%\label{F:SubstructureConeStability}
%\end{figure}

Another example of when the stability decreases is if $P$ is a simplex with stability 1. Then a cone on any proper subcomplex will necessarily decrease the stability.

These examples indicate that further investigation is needed of the case $G<k+1$. One consideration  is to perhaps restrict attention to subcomplexes in which there is a genuine incompatibility. This means that a cone would only be  taken over a subcomplex that is not itself a simplex. Another distinction that appears to make a difference is whether a cone is taken over simplices that belong to different components of $P$. While bringing different components together via a mediator seems desirable and should increase stability, the example above with $G=k+1$ indicates that this is not necessarily true. However, it may still be true for $G<k+1$, and this would make that edge case all the more interesting.

%%%%%%%%%%%%%%%%%%%%%%%%%%%%%%%%%%%%%%%%%%%%%%%%%%%%%%%%%%%%%%%%%%%%%%%%%%%%%%%%%%%%%%%%%%%%%%%%%%%%%%%%%%%%%%%%%

\subsection{Weighted simplicial complex model}\label{S:SimplicialModelIssues}

%%%%%%%%%%%%%%%%%%%%%%%%%%%%%%%%%%%%%%%%%%%%%%%%%%%%%%%%%%%%%%%%%%%%%%%%%%%%%%%%%%%%%%%%%%%%%%%%%%%%%%%%%%%%%%%%%

One shortcoming of the simplicial complex model so far is that it does not distinguish the level to which agents are compatible or on which issues they agree.  For example, an edge between two agents might mean that they are like-minded on one or all the issues, but the simplicial complex does not capture this difference.  In this section, we offer a refinement of some of the ideas introduced so far and indicate some possible further directions of investigation. We extrapolate to the situation where the agents are not asked to decide on a single item but rather have a preference on a collection of issues or need to decide which facets of an issue they deem to be essential to the functioning of the underlying system.

First suppose $P=(A,\calC)$ is a political structure as before, but now consider the situation where each viable configuration $\gamma\in\calC$ has a \emph{weight} $w_\gamma\in\R$ associated to it. Alternatively, a viable configuration is now considered to be a pair $(\gamma, w_\gamma)$. 

\begin{defin}\label{D:WeightedStructure}
A \emph{weighted political structure} $P$ is a triple $(A, \mathcal C,W)$, where $A$ and $\mathcal C$ are a finite set of agents and a collection of viable configurations as in \refD{PoliticalStructure}, and $W$ is a set of weights associated to the elements of $\mathcal C$.
\end{defin}

Now suppose we are given a set $I=\{p_1, ..., p_m\}$ of $m$ \emph{issues} that the agents $a_0, ..., a_k$ must consider. We will call the set $I$ an  \emph{agenda}.  One way to associate a weighted political system to this situation is to declare a configuration $\gamma$ to be viable if agents in it agree on at least one issue, and the weight $w_\gamma$ to be the number of issues the agents in $\gamma$ agree on. We also set $w_{a_i}=m$ for all $0\leq i\leq k$; this convention reflects the fact that each agent agrees with themselves on all the issues.

Under this setup, we will say that $\gamma$ is \emph{$w_\gamma$-viable} (we will also just write \emph{$w$-viable} if the underlying configuration is understood).
 We will say that $P$ is \emph{fully $w$-viable} if all agents agree on $w$ issues. When $w=m$, then $P$ is \emph{fully viable}.

Now we can build a simplicial complex exactly the same way as before.  For each $w$-viable configuration $\gamma$, $w\geq 1$, the simplicial complex contains a $(|\gamma|-1)$-dimensional simplex, only now a weight is associated to it. In the abstract simplicial complex, we denote the weight by a subscript, and in its realization, we simply label the simplex.

This simplicial complex can also be built in the following way: For $0\leq i< j\leq k$, let 
$$
\vec v_{ij} = (v_1, ..., v_m),
$$
where
$$
v_l = \begin{cases}
1, & \text{if agents $a_i$ and $a_j$ agree on issue $p_l$};\\
0, & \text{if agents $a_i$ and $a_j$ disagree on issue $p_l$.}
\end{cases}
$$
We will call $\vec v_{ij}$ the \emph{agreement vector} of $a_i$ and $a_j$.

The agreement vectors determine all the weights since a configuration $\gamma$ is $w$-viable if the vectors $\vec v_{ij}$ have $w$ common 1's for all pairs of agents $a_i$ and $a_j$, $i\neq j$, in the configuration.

%When we need to be clear about the underlying configuration, we will denote this number by $w_\gamma$ and will also refer to it as the \emph{weight of $\gamma$}. 

Note that, for a configuration of size 2 with agents $a_i$ and $a_j$, 
$$
w_{\{a_i, a_j\}}=\sum_{i=1}^m v_i,
$$
the sum of the components of the agreement vector $\vec v_{ij}$. 

Also note that, if $\vec v_{ij}$ and $\vec v_{ij'}$, $j<j'$, have a common 1 in some coordinate, so does $\vec v_{jj'}$. Same for $\vec v_{ij}$ and $\vec v_{i'j}$, $i<i'$,  with $\vec v_{ii'}$ also having the common 1 in the same coordinate. This simply reflects the fact that, given three agents, if two pairs of them agree on an issue, then all three do. Thus for configurations of size 3 with agents $a_i$, $a_j$, and $a_k$, the weight is simply the dot product of any two of the three possible vectors, i.e.
$$
w_{\{a_i, a_j, a_k\}}=\vec v_{ij}\cdot \vec v_{ik}.
$$

%Now let the \emph{weight} of $C$, denoted by $w_C$, be the number of common 1's in all the vectors $v_{ij}$ corresponding to agents $a_i$ and $a_j$ in $C$. 

\begin{example}\label{Ex:WeightedStructure}
Suppose $P$ consists of five agents and there are nine issues on the agenda.  Suppose the non-zero agreement vectors are

\begin{align*}
\vec v_{01} & = (0, 1, 0, 0, 0, 0, 0, 0, 0 ) \\
\vec v_{02} & = (0, 0, 0, 0, 0, 0, 1, 1, 0 ) \\
\vec v_{03} & = (1, 0, 0, 0, 1, 1, 1, 0, 0 ) \\
\vec v_{13} & = (0, 0, 1, 0, 0, 0, 0, 0, 0 ) \\
\vec v_{14} & = (1, 0, 0, 1, 0, 0, 0, 0, 0 ) \\
\vec v_{23} & = (0, 0, 0, 0, 0, 0, 1, 0, 1 ) \\
\end{align*}

%\begin{align*}
%\vec v_{01} & = (0, 1, 0, 0, 0, 0, 0, 0, 0 ) & \vec v_{13} & = (0, 0, 1, 0, 0, 0, 0, 0, 0 ) \\
%\vec v_{02} & = (0, 0, 0, 0, 0, 0, 1, 1, 0 ) & \vec v_{14} & = (1, 0, 0, 1, 0, 0, 0, 0, 0 ) \\
%\vec v_{03} & = (0, 0, 0, 0, 1, 1, 1, 0, 0 ) & \vec v_{23} & = (0, 0, 0, 0, 0, 0, 1, 0, 1 )
%\end{align*}

Then the weighted simplicial complex for $P$ is 
$$
\{ 
\{a_0, a_2, a_3\}_1, \{a_0, a_2\}_2, \{a_0, a_3\}_4, \{a_2, a_3\}_2, \{a_0, a_1\}_1, \{a_1, a_3\}_1, \{a_1, a_4\}_2, 
\{a_0\}_9, \{a_1\}_9, \{a_2\}_9, \{a_3\}_9, \{a_4\}_9    
\}
$$

Its geometric realization is given in Figure \ref{F:Weighted}.

\begin{figure}[h]
\centering
\includegraphics[width=1.7in]{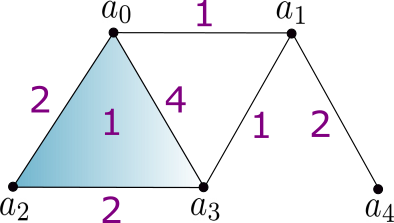}
\caption{An example of a weighted political structure. Each vertex has weight 9, but we have omitted labeling this in the picture to preserve its clarity.}
\label{F:Weighted}
\end{figure}

\end{example}

We then have natural weighted analog of \refD{Stability}.

\begin{defin}\label{D:WeightedStability}
Given a weighted political structure $P=(A,\mathcal C, W)$ with the set of agents $A=\{a_0, ..., a_k\}$ and agenda $I$ of cardinality $m$, define the \emph{weighted viability of $a_i$} to be 
\begin{equation}\label{E:WeightedAgentStability}
\operatorname{wvia}(a_i)=\frac{1}{m(2^k-1)}\left(\sum_{C\in \operatorname{st}(a_i)}w_C-m\right).
\end{equation}

Define the \emph{weighted stability of $P$} to be

\begin{equation}\label{E:WeightedComplexStability}
\operatorname{wstab}(P)=\frac{1}{m(2^{k+1}-(k+2))}\left(\sum_{C\in \mathcal C}w_C-m(k+1)\right).
%\operatorname{stab}(P)=\frac{1}{2^{k}(k+1)}\left( \sum_{a_i\in A}|\operatorname{st}(a_i)| -(k+1) \right)
\end{equation}

\end{defin}

Weighted viability is again the number of simplices that have $a_i$ as a vertex, but weighted and normalized. This is minimized when the vertex is isolated (and it carries the weight $m$ by convention), and maximized when $a_i$ is the vertex of the greatest possible number of simplices, $2^k$, each with weight $m$. Similarly for the weighted stability which is minimized when each of the $k+1$ vertices is isolated and maximized when each subset of the full $k$-simplex has weight $m$.

\begin{example}
In the weighted structure from Example \ref{Ex:WeightedStructure}, we have $\operatorname{wvia}(a_0)=8/135$, $\operatorname{wvia}(a_1)=4/135$, $\operatorname{wvia}(a_2)=5/135$, $\operatorname{wvia}(a_3)=8/135$, $\operatorname{wvia}(a_4)=2/135$, and $\operatorname{wstab}(P)=1/18$. (Without any information about the agreements on particular issues or weights, the stability of the system, as defined in \refS{Stability}, would be $\operatorname{stab}(P)=7/26$.)
\end{example}

Note that, when $|I|=1$, i.e.~there is only one issue on the agenda, we recover the setup from \refS{SimplicialModel}, including  \refD{Stability}. The single issue can be thought of as general compatibility, and agreement on this issue is simply a question of whether two agents can coexist in a political system.

It is not difficult to see that the results from \refS{PoliticsMods}, in particular Propositions \ref{P:WedgeStability}  and \ref{P:ConeStability}, readily carry over to the setting of weighted systems. We leave the details to the reader.

There are several potential further directions that our definition of stability could be employed in.

\begin{itemize}
\setlength\itemsep{4pt}

\item The situation where a simplicial complex is labeled by values is familiar from the field of \emph{signal processing}. Signals on simplicial complexes have been used to model various phenomena \cite{BS:SimplicialSignals, BSC:SimplicialSignals,  GGB:NeuroSimplex, JKP:SimplicialSignals, MDBRS:Collaboration}. Namely, given a $d$-complex $P$ modeling a political system, let $P(n)$ be the collection of its $n$-simplices (but not their faces, as is the case in the $n$-skeleton $P^{n}$). A \emph{signal} over $P(n)$ is a function 
$$
s^n\colon P(n)\longrightarrow \R,
$$
and the collection $\{s^0, s^1, ..., s^d\}$ is a \emph{signal over $P$}. There are a number of standard quantities and procedures that can be associated with signals on simplicial complexes. For example, one can analyze the Laplacians of $P$ mentioned in \refS{Homology} and use their \emph{Hodge decomposition} along with the signal to define the \emph{Fourier transform of $P$}. In the context of political structures, this can likely be used to extract qualitative information about the signal, such as measure the average degree of disagreement of an agent compared to its neighbors (something along these lines was done in \cite{HG:Discourse}) or even define a more nuanced notion of stability of a political system.

\item Instead of adding the values of the coordinates of the agreement vector $\vec v_{ij}$, one could retain them and associate to all the simplices the ``intersection vector'' with common agreement coordinates being 1. This vector can also encode the situation where the agents have different voting power.  

\vspace{4pt}
\noindent
This setup can be used to define a generalized power index of a weighted voting system with coalition restrictions. Namely, the simplex structure encodes all possible coalitions (and additionally says which coalitions are not possible via the missing simplices) on particular issues.  Given a quota $q$ on each issue, one should be able to define and calculate a cumulative Shapley-Shubik-style power index for each agent over all the issues. The existence of simplices of certain dimensions, taken with the sum of the weights of the agents that are its vertices, would indicate a winning coalition. The link of an agent could be used to determine whether the agent is pivotal for the coalition. Finally, the order of how the coalitions are formed could be governed by the structure of dominated vertices. For the setting of a single issue, this should recover the usual Shapley-Shubik index (or at least something closely related to it) but enhanced to  take into account forbidden coalitions.

\item The entries of the agreement vector can even be generalized to reflect the level of agreement on the $l$th issue. In other words, instead of $v_l\in\{0,1\}$, one could have $v_l\in [0,1]$ with 0 as before representing complete disagreement and 1 complete agreement. The value associated to a simplex of dimension greater than 1 would be the minimum of the pairwise agreement values of its vertices. 

\vspace{4pt}
\noindent
This turns out to be the setting that is even more general than that of signal processing, namely that of \emph{sheaves} over simplicial complexes \cite{Robinson:Networks, Robinson:Circuits, Robinson:TSP}. To each simplex of $P$, one now associates the vector space $\R^m$ (where $m$ is the number of issues). Following \cite{HG:Discourse}  (see also \cite{vDGLR:KnowledgeSimplicial}), these can be thought of as opinion spaces of the agents or subsets of agents. Inclusion of a face corresponds to a linear transformation which can be interpreted as a discourse space. This data defines a \emph{sheaf over $P$}.  Attaching agreement vectors to the simplices of $P$ is a \emph{section} of the sheaf.  Using corresponding Laplacians and \emph{sheaf cohomology}, it should be possible to study the dynamics of the political system in which agents are allowed to change their opinions.

\item A coordinate $v_l\in [0,1]$ could also represent the probability that two agents agree on the $l$th issue. This could serve to produce an a priori expected stability of the system. Furthermore, one could study probability-based Markov chains on $P$ \cite{MS:RandomSimplicial, PR:RandomSimplicial, SBHLJ:RandomSimplicial, WZLS:SocialSimplicial}. Random walks might represent coalitions merging with other coalitions according to the probability that can be extracted from the agreement probability vectors $\vec v_{ij}$. It is likely that this would tie together nicely with existing approaches to modeling conflict dynamics via Markov chains \cite{DS:Conflict, HNWRC, Shallcross:Conflict, Shearer:Conflict}.

\item Finally, one could take into account changing preferences, attitudes, or opinions of the agents over time. This \emph{evolutionary} model can even be applied to weighted complexes \cite{CB:GrowSimplicial, HK:VoteSimplicial, SMSS:WeightedSimplicial} and could potentially be used to quantify the stability of political structures over time and track the influence of mediators or the development of coalition patterns. This could also be combined with some of the above considerations; for example, one could define a time-series power index of an agent over an issue or a set of issues.

\end{itemize}

%%%%%%%%%%%%%%%%%%%%%%%%%%%%%%%%%%%%%%%%%%%%%%%%%%%%%%%%%%%%%%%%%%%%%%%%%%%%%%%%%%%%%%%%

\subsection{Homology and political structures}\label{S:StructuresHomology}

%%%%%%%%%%%%%%%%%%%%%%%%%%%%%%%%%%%%%%%%%%%%%%%%%%%%%%%%%%%%%%%%%%%%%%%%%%%%%%%%%%%%%%%%

In this section, we indicate some ways in which homology interacts with political structures. 

To begin, we demonstrate how homology captures ``cycles of non-viabilities'' in political structures. 
Recall from \refS{Homology} that the Betti numbers $\beta_n$ encode the ranks of the homology groups of a simplicial complex. Also recall that an $n$-cycle of simplices is a collection of $n$-simplices that has no boundary.

\begin{thm}\label{T:StructuresHomology} Let $P=(A,\mathcal C)$ be a political structure. Suppose $\beta_n\neq 0$ for some $n\geq 0$. Then $P$ contains $\beta_n$ subsets of agents, each of cardinality at least $n+2$, which do not form viable configurations. In particular, $P$ is not fully viable, i.e.~agents are not all compatible. In the weighted model, these $\beta_n$ subsets of agents are not $w$-viable for any $w$. In particular, $P$ is not $w$-viable for any $w$.
%\begin{enumerate}
%\item $P$ is not fully viable, i.e.~agents are not all compatible. In the weighted system model, this means that $P$ is not $w$-viable for any $w$, i.e.~there is no issue that all agents agree on.
%\item More precisely, $P$ contains $\beta_n$ subsets of agents $n$-cycles in which no $n+2$ agents form a viable configuration. In the weighted model, no $n+2$ of these $n$-cycles of agents are $w$-viable for any $w$.
%\end{enumerate}

\end{thm}

We will call each subset of agents detected by $\beta_n$ an \emph{$n$-cycle of non-viability}.

Before we prove this result, we illustrate it with an example.

\begin{example}\label{Ex:StructuresHomology}
It is not hard to calculate that the Betti numbers for the political structure in Figure \ref{F:PoliticalStructureHomology} are $\beta_0=0$, $\beta_1=2$, $\beta_2=1$. Corresponding to $\beta_1=2$, there exist two subsets of agents, $\{ a_2, a_3, a_4, a_5\}$ and $\{ a_5, a_6, a_7\}$, that are not viable. Corresponding to $\beta_2=1$, the subset $\{ a_0, a_1, a_2, a_3\}$ is not viable.

\begin{figure}[h]
\centering
\includegraphics[width=2.2in]{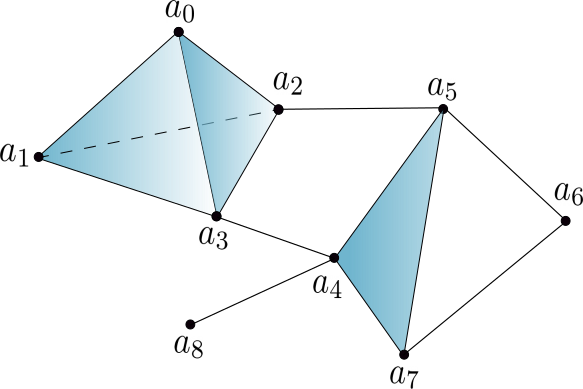}
\caption{A political structure $P$ with some cycles of non-viability. The tetrahedron is meant to be hollow.}
\label{F:PoliticalStructureHomology}
\end{figure}

\end{example}

\begin{example}\label{Ex:MiddleEastBetti}
Looking back at \refEx{MiddleEast}, we have that $\beta_0=2$ and $\beta_1=1$. The three components are the main complex and the two isolated agents $a_{11}$ and $a_{12}$. Corresponding to $\beta_1=1$, we have a 1-cycle of non-viability $\{ a_2, a_3, a_7, a_8\}$. This is not a unique cycle since $\{ a_2, a_3, a_7, a_5, a_8\}$ or $\{ a_2, a_3, a_7, a_{10}, a_8\}$ could also be taken as the cycles representing the single 1-dimensional homology class.
\end{example}

 \begin{proof}[Proof of \refT{StructuresHomology}] 
 If $\beta_n>0$ then one can choose $\beta_n$ $n$-cycles which generate $\Ho_n(P)$. Each of these cycles must contain at least $n+2$ agents. To see this, recall that an $n$-cycle is a formal linear combination of $n$-simplices. An $n$-simplex requires at least $n+1$ agents, so there cannot be fewer. If there are precisely $n+1$ agents, then we have a cycle consisting of a single $n$-simplex, which cannot be a generator of nontrivial homology of any degree since every simplex is acyclic.
 
Each of these collections of agents cannot be viable since, if it were, it would form a simplex.  Since a simplex is acyclic, this cycle would be homologous (differ by a boundary) to the constant cycle and could thus not be a homology generator.  
 \end{proof}

\begin{rems}\ 

\begin{itemize}
\setlength\itemsep{4pt}

\item There might be sets of agents that satisfy the conclusions of \refT{StructuresHomology} but are not detected by homology. For example, the 1-cycle determined by the set of vertices $\{ a_0, a_1, a_2, a_3\}$ that goes around the perimeter of the square in Figure \ref{F:HomologyNotComplete} is not viable, but all the Betti numbers of that structure are zero.

\begin{figure}[h]
\centering
\includegraphics[width=1in]{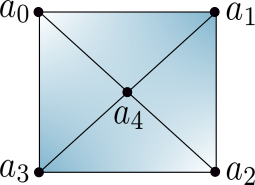}
\caption{A political structure with trivial homology.}
\label{F:HomologyNotComplete}
\end{figure}

\item Replacing a homology generator by another, i.e.~by a homologous one, would give a new subset of non-viable agents. Thus the subsets of agents predicted by the Betti numbers are not uniquely determined by homology. This was pointed out in \refEx{MiddleEastBetti}.

\item The subsets of non-viable agents identified in \refT{StructuresHomology} are not necessarily distinct. Homology generators can meet along common subcomplexes and can thus share agents. 

\end{itemize}
\end{rems}

There is a little bit more that can be extracted from the Betti numbers beyond the statement of \refT{StructuresHomology}. For instance, in Example \ref{Ex:StructuresHomology}, we have two groups of four non-viable agents,  $\{ a_2, a_3, a_4, a_5\}$ and $\{ a_0, a_1, a_2, a_3\}$. However, the fact that the latter group forms a cycle that is detected by $\Ho_2$ means that some of its subsets of size three are compatible.  This is because $2$-chains are formed by $2$-simplices, so various subsets of this cycle must form triple viabilities. Therefore this subset of agents is in some sense ``closer'' to being viable than the subset $\{ a_2, a_3, a_4, a_5\}$. Alternatively, the defect of the cycle formed by $\{ a_0, a_1, a_2, a_3\}$ is less than the defect of the cycle formed by $\{ a_2, a_3, a_4, a_5\}$.

This observation can be used to selectively introduce mediators into political substructures to reduce the defect and increase stability in the most efficient way. Namely, suppose it is not feasible to mediate the entire political structure and the efforts instead have to be aimed at mediating subsets of agents. Topologically, this corresponds to coning off a subcomplex of $P$. Introducing a mediator into a situation where all agents are already viable, i.e.~constructing a cone on a simplex to produce another simplex of a higher dimension, is a waste of resources. Homology detects this, since a simplex is acyclic. 

At the other extreme, if a system has disconnected subsets of agents and is thus dysfunctional, selecting and connecting particular agents from each component for mediation might be the only way to move forward. This corresponds to constructing a cone on vertices from distinct components. Homology detects the improvement since, in the initial structure, $\Ho_0(P)\neq 0$, but after the mediation, $\Ho_0(P)= 0$. 

In between are the situations where coning off cycles that represent elements of $\Ho_n(P)$ removes incompatibilities and increases stability, starting with the smallest $n$ for which homology is nontrivial.

Furthermore, \refT{StructuresHomology} could have been stated in terms of cores.  The core $P^c$ of $P$ is the substructure in which no friendly delegations are possible (see \refS{Collapse}). Since $P$ is strongly homotopy equivalent to its core, we have 
\begin{equation}\label{E:CoreHomology}
\Ho_n(P)=\Ho_n(P^c).
\end{equation}
Since a core is a minimal representation, its homology cycles are minimal obstructions to viability. A targeted introduction of a mediator to those cycles is thus the most efficient way to turn $P$ into a compatible structure.

Homology can also detect the lack of the presence of a mediator. 

\begin{prop}\label{P:NoMediatorHomology}
Let $P=(A,\mathcal C)$ be a political structure. Suppose $\beta_n\neq 0$ for some $n\geq 0$. Then there is no mediator in $P$. 
\end{prop}

\begin{proof}
If $P$ has a mediator, then $P=CK$, the cone on some complex $K$. Since the cone is always (strongly) contractible, it has trivial homology (see \refS{Homology}).
\end{proof}

If two mediators are introduced, then, as we know, the model for this is the suspension $\Sigma P$ of $P$. We also know from the comments at the end of Example \ref{E:ConeSuspension} and the last part of \refP{HomologyOfOperations} that the homology of $P$ and $\Sigma P$ are just shifts of each other. It immediately follows from this and \refT{StructuresHomology} that, if $P$ has an $n$-cycle of non-viability in $P$ (as detected by $\beta_n$), then $\Sigma P$ will contain an $(n+1)$-cycle of non-viability. If more than two mediators are introduced, then, for each pair of mediators, a new $(n+1)$-cycle of non-viability will be introduced in $\Sigma P$.

It is not clear how to extend the results of this section to weighted complexes. There is a version of homology for this situation \cite{D:WeightedComplex, WRWX:WeightedMorse} that could generalize our observations to weighted complexes, but it requires the divisibility of the weights with respect to face inclusions.

%%%%%%%%%%%%%%%%%%%%%%%%%%%%%%%%%%%%%%%%%%%%%%%%%%%%%%%%%%%%%%%%%%%%%%%%%%%%%%%%%%%%%%%%

\section{Delegations and compromises}\label{S:Delegations}

%%%%%%%%%%%%%%%%%%%%%%%%%%%%%%%%%%%%%%%%%%%%%%%%%%%%%%%%%%%%%%%%%%%%%%%%%%%%%%%%%%%%%%%%

%
%
%
%The following shows an example of a political structure with four different parties. In one case, we have parties that all agree or have at least some commonality with each other. In the second case, with disjoint parties, no concensus is possible. 
%
%\includegraphics[height=5in]{political structure example.pdf}
%\newline In the language of algebraic topology, a political structure is a simplicial complex, where E describes the vertices and K the simplexes in E. 

The motivating setup in which Abdou and Keiding \cite{AK:Conflict} introduce the simplicial complex model for political structures is that of agents delegating, or relinquishing power to other agents that are better positioned than them in the structure. Upon delegation, the delegating agent is removed from the structure. 

More precisely, an agent $a$ is better positioned than agent $b$ in the political structure $P$ if vertex $b$ is dominated by $a$ in the simplicial complex representing $P$. One in this case also says that $a$ is more \emph{central} than $b$.  A \emph{friendly delegation} from $b$ to $a$, denoted by $\delta_{b\to a}$ \cite[Definition 2]{AK:Conflict}, is the deletion of $b$ from $P$. This corresponds to an elementary strong collapse $P\, \esc\, P\setminus b$ (see \refD{ElemStrongCollapse}).

The following result relates the notion of a friendly delegation to that of viability and stability. Suppose, as before, that $f_n$ is the number of viable configurations of $P$ with $n$ agents.  As usual, let $\lk(b)$ be the link of vertex $b$ in the simplicial complex representing $P$.

\begin{thm}\label{T:DelegationStability} Suppose $P$ is a $d$-dimensional political structure with $k+1$ agents and agent $b$ is dominated by agent $a$. 
\begin{enumerate}
\item A friendly delegation from agent $b$ to agent $a$ increases the viability of agent $a$ if and only if there exists a viable configuration containing $a$ but not $b$.
\item If $\sum_{i=0}^d f_i-|\lk(b)|\geq k+1$, then a friendly delegation increases the stability of $P$.
\end{enumerate}
\end{thm}

\begin{proof}
For (1), suppose $b$ is dominated by $a$. We have
$$
|\operatorname{st}_P(a)|=\underbrace{\sum_{\substack{\sigma\ \text{maximal in $P$} \\ a,b\in\sigma}}2^{\dim\sigma}
-\sigma\text{-overcount}}_S +
\underbrace{\sum_{\substack{\tau\ \text{maximal in $P$} \\ a\in\tau, b\notin\tau}}2^{\dim\tau}
-\tau\text{-overcount}}_T
%\ -1\ 
$$

where $\sigma$-overcount is the number of faces that are overcounted in the first sum; this is due to the fact that some of the maximal simplices $\sigma$ in that sum might share common faces. Similarly for $\tau$-overcount. To simplify the notation, the two summands will be abbreviated by $S$ and $T$, as indicated by the braces under the equation.

The hypothesis that there exists a viable configuration containing $a$ but not $b$ precisely means that $T>0$.

Denote by $\operatorname{via}_P(a)$ the viability of $a$ before the delegation and by $\operatorname{via}_{P\setminus b}(a)$ the viability of $a$ after the delegation. Similarly denote by $\operatorname{st}_P(a)$ and $\operatorname{st}_{P\setminus b}(a)$ the stars of $a$ before and after the delegation. We want to show that 
$$
\operatorname{via}_P(a) < \operatorname{via}_{P\setminus b}(a),
$$
or
\begin{equation}\label{E:FriendlyStability}
\frac{1}{2^k-1}(|\operatorname{st}_P(a)|-1)< \frac{1}{2^{k-1}-1}(|\operatorname{st}_{P\setminus b}(a)|-1).
\end{equation}

%The subtraction of $1$ at the end is due to the overcount of agent $a$ which appears once in both $S$ and $T$.

Note that $\dim\sigma\geq 1$ for all maximal $\sigma$ containing $a$ and $b$. After the friendly delegation (i.e.~elementary collapse), each $\sigma$ is reduced in dimension by one, so that 
$$
|\operatorname{st}_{P\setminus b}(a)|
%=\frac{1}{2}\sum_{\substack{\sigma\ \text{maximal in $P$} \\ a,b\in\sigma}}2^{\dim\sigma} +
%\sum_{\substack{\tau\ \text{maximal in $P$} \\ a\in\tau, b\notin\tau}}2^{\dim\tau}
%-\text{overcount}_{P\setminus b}
=\frac{S}{2}+T
$$
Then verifying \eqref{E:FriendlyStability} comes down to verifying
$$
2^{k-1}<\frac{S}{2} + 2^{k-1}T.
$$
Since $T>0$ by assumption (and $S>0$ always), this is true.  This is also a necessary condition since, if $T=0$, we have $S\leq 2^k$ since the greatest cardinality of a star that a vertex can have in a complex of $k+1$ vertices is $2^k$, and that happens precisely when the complex is a $k$-simplex.

For (2), we use a standard result that the closed star of a vertex is a cone over its link. Since the friendly delegation from $b$ to $a$ (i.e.~elementary collapse) removes $\st(b)$, the inverse of a friendly delegation (i.e.~elementary expansion) is therefore a cone on $\lk(b)$. Now since the number of simplices not in $\lk(b)$  is by assumption greater than $k+1$, \refP{SubstructureConeStability} says that such an elementary expansion decreases the stability of $P$.  Thus the elementary collapse from $b$ to $a$ must increase the stability.
\end{proof}

Note that the end of the proof of (1) also shows that the viability of an agent remains unchanged after that agent has been delegated to corresponds precisely to the situation when the agent is maximally viable, namely when the system is a $k$-simplex. A delegation then reduces the system to a $(k-1)$-simplex in which the agent remains maximally viable.

%The interpretation of (1) is that ...

Since there could be more than one vertex $a$ that dominates $b$ (for example, in the left picture of Figure \ref{F:PoliticalStructures}, $a_2$ is dominated by $a_1$ and $a_3$), a friendly delegation of $b$ to $a$ simultaneously changes the viability of all the vertices dominating $b$ since an elementary collapse that deletes $b$ results in the same complex regardless of how many or which vertices dominate it.  This reflects the fact that $b$ could have delegated to any of those more central agents and the effect on the political structure would have been the same. 

An interesting feature is thus that some viabilities might increase and some might decrease depending on whom the power is delegated to, depending on whether the dominating vertices are members of viable configurations that include $b$ or not. For example, in the left picture of Figure \ref{F:PoliticalStructures}, deleting $a_2$ decreases the viability of $a_3$ but increases the viability of $a_1$.

These considerations indicate that the definition of a friendly delegation and the definition of viability could possibly be amended to take into account not just the delegating agent but also the one that the power is being delegated to.

%That a friendly delegation would decrease the stability is not surprising since an agent who belonged to some viable configuration is removed along with those configurations.  
%Removing an agent... 

We plan to investigate the situation when $\sum_{i=0}^d f_i-|\lk(b)|< k+1$ in the near future. The conditions here are more restrictive than those in  \refP{SubstructureConeStability} since an elementary expansion is now not a cone on an arbitrary subcomplex but on $\lk(b)$. In turn, by \refP{DominatedLink}, this link is itself a cone on $a$. In other words, a friendly delegation is the inverse of introducing a mediator on a subcomplex that is a cone on a vertex, and, under this condition, \refP{SubstructureConeStability} might have a more straightforward extension to the case  $\sum_{i=0}^d f_i-|\lk(b)|< k+1$.

Reducing a structure by a delegated removal of agents to one where all that is left is a viable configuration is a \emph{represented compromise}, or \emph{$R$-compromise}  \cite[Definition 9]{AK:Conflict}. If iterated delegations are needed, i.e~an agent that has been delegated to proceeds to delegate to a third agent, and this process possibly continues, ending in a viable configuration, then this is a \emph{delegated compromise} or \emph{$D$-compromise} \cite[Definition 12]{AK:Conflict}.

Two of the central results in \cite{AK:Conflict} are the following.

\begin{prop}[\cite{AK:Conflict}, Propositions 6 and 10]\label{P:AbdouKeidingCompromises}
If a political structure has an $R$-compromise, then it is strongly contractible. A political structure has a $D$-compromise if and only if it is strongly contractible.
\end{prop}

We can relate the notion of a $D$-compromise to that of mediators and merging of structures.

\begin{prop}\label{P:WedgeMediatorCompromise}
Suppose $P$ is a political structure.
\begin{enumerate}
\item If $P$ contains a mediator, then $P$ has a $D$-compromise.

\item If $P$ contains more than one mediator, then $P$  has a $D$-compromise if and only if the structure that remains once mediators are removed has a $D$-compromise.

\item If $P$ is the result of two political structures $P'$ and $P''$ merging, then $P$ has a $D$-compromise if and only if $P'$ and $P''$ have $D$-compromises.
\end{enumerate}

\end{prop}

\begin{proof}
Part (1) follows from the first part of \refC{ConeSuspensionStrong}. Part (2) follows from the second part of \refC{ConeSuspensionStrong}, but generalized from suspension to join with an arbitrary number of vertices. Part (3) follows from \refC{WedgeCollapse}.
\end{proof}

We can also use homology to detect the nonexistence of compromises.

\begin{prop}\label{P:HomologyCompromise}
Suppose $P$ has nontrivial homology, i.e.~$\beta_n\neq 0$ for some $n\geq 0$. Then $P$ does not have an $R$-compromise or a $D$-compromise.
\end{prop}

\begin{proof}
Nontrivial homology means that $P$ is not strongly contractible (see \refS{Homology}). By \refP{AbdouKeidingCompromises}, this means that $P$ does not have an $R$-compromise or a $D$-compromise.
\end{proof}

One could again regard these results in terms of the core $P^c$. The minimal complex obtained by delegated compromises, i.e.~by iterated strong elementary collapses, is in the terminology of Abdou and Keiding \cite{AK:Conflict} called a \emph{delegated core} or a \emph{$D$-core}. If the $D$-core of $P$ is a single vertex, then a $D$-compromise exists. As was observed in \eqref{E:CoreHomology}, the homology groups of $P$ and $P^c$ are isomorphic.  If the homology of the $D$-core is not trivial, then it follows that $P^c$ is not a single vertex and $P$ therefore does not have a $D$-compromise. (This observation is essentially the homological version of a part of \cite[Proposition 10]{AK:Conflict}.)

The $D$-core is the minimal subcomplex that contains homological obstructions to a $D$-compromise and it thus gives the most efficient prescription for capping off the nontrivial cycles using mediators so as to increase the potential of a compromise in the structure.

\begin{example}\label{Ex:MiddleEastCompromise}
Continuing with the Middle East example, \refEx{MiddleEast}, we observed in \refEx{MiddleEastBetti} that this simplicial complex has nontrivial homology.  According to \refP{HomologyCompromise}, this structure therefore does not have an $R$-compromise or a $D$-compromise. Its $D$-core consists of 0-cycles $a_{11}$ and $a_{12}$ and a 1-cycle $\{ a_2, a_3, a_7, a_8\}$. The most efficient mediations would thus be ones that cone off some subsets of these six agents.
\end{example}

It would be interesting to extend the notion of a friendly delegation in the context of weighted political structures. One drawback of the current definition is that the identity of an agent is lost when they delegate to another one.  If the agents carried some weight, then the delegating agent's weight could be combined with that of the agent who has been delegated to. This would resemble a more standard process of the formation of coalitions, but with the caveat that the structure of dominated vertices would put restrictions as to the order in which these coalitions could be formed via delegations. This process could be aggregated over all the issues, i.e.~over the weights of viable configurations so as to produce a compounded power index over the entire issue space (something along these lines was already alluded to at the end of \refS{SimplicialModelIssues}).

\section{Future directions}\label{S:Future}

%%%%%%%%%%%%%%%%%%%%%%%%%%%%%%%%%%%%%%%%%%%%%%%%%%%%%%%%%%%%%%%%%%%%%%%%%%%%%%%%%%%%%%%%

One direction we have not explored in this paper is to potentially utilize the  notion of the \emph{poset} of a simplicial complex $K$ consisting of faces of $K$, ordered by inclusion (for details, see, for example \cite{W:PosetTopology}).  In addition to providing another model and calculational tool for the viability of an agent and the stability of the political system, there are other potential uses for the poset approach. Namely, the lengths of the chains in this poset could provide insights into the structure of maximal coalitions and the defect. One could give interpretation to the notion of a \emph{pure} political structure (all maximal coalitions have the same dimension). The \emph{M\"obius function} of a poset could also prove to be fruitful since it is related to the Euler characteristic. The Euler characteristic is in turn the alternating sum of the Betti numbers, and we have seen that those are relevant for analyzing political structures. The functoriality of the poset construction could also be brought to bear.

Another important construction that can be performed on simplicial complexes is that of a \emph{barycentric subdivision}. As was mentioned in Remark \ref{R:BarycenterContiguity}, barycentric subdivision is intertwined with the notion of contiguity which is the central idea in \cite{AK:Conflict} and in our \refS{Delegations}. A barycentric division of a simplex corresponds to introducing a new agent who breaks up a viable configuration of $n$ agents and creates $n$ new ones consisting of subsets of original agents of size $n-1$, with the new agent present in all of them.  It would be interesting to explore how barycentric subdivisions interact with friendly delegations. In addition, the poset of a complex and the barycentric subdivisions are related since the poset functor can be followed by the functor that creates a simplicial complex essentially by turning chains into simplices, and the composition of the two is precisely the barycentric subdivision of the original complex. Finally, barycentric subdivision is needed to endow the product of complexes $K\times L$ with the structure of a complex, so bringing this into the picture might give meaning to the notion of a product of political structures.

Given a field $\mathbf k$, consider $\mathbf k[v_0, ..., v_k]$, the polynomial ring on the vertices of a simplicial complex $K$. Let $\mathcal I$ be the ideal generated by the non-faces of $K$, namely those in the complement of $K$ in the full simplex determined by the vertices of $K$. This is the well-known \emph{Stanley-Reisner ideal}. The \emph{Stanley-Reisner ring} associated to $K$ is then the quotient $\mathbf k[v_0, ..., v_k]/\mathcal I$. It would be useful to interpret algebraic invariants, in particular its homology (which can be computed using Gr\"obner bases), of this ring in terms of political structures. A simplicial complex is uniquely determined by its Stanley-Reisner ring \cite{BG:StanleyReisner}, so anything one can learn about political structure from its simplicial complex model can also be deduced from the associated Stanley-Reisner ring. 

A more classical notion of a collapse dates back to J.H.C.~Whitehead \cite{W:Simplicial}. According to this, if $\tau$ is a face of $\sigma$ and $\sigma$ is the only simplex that has $\tau$ as a proper face, then a collapse removes all simplices  contained in $\sigma$ and containing $\tau$ (including $\tau$ and $\sigma$ themselves). The notion of a strong collapse is stronger than this usual collapse \cite[Proposition 2.14]{BM:StrongHtopy}. It would be interesting to see which segments of our work here extend to the setting of usual collapses. One advantage of using strong collapses is that 
\refP{JoinCollapse} is not known for collapses (see comment immediately preceding Proposition 2.15 in \cite{BM:StrongHtopy}). Strong collapses are also tailor-made for friendly delegations, but a more general idea of coalition delegation (an entire coalition delegates its power to an agent) might be the natural generalization modeled by the usual collapses.

%One reason that it's nice to use strong collapse is the comment in Barmak-Minian right before their Prop 2.15, which is our \refP{JoinCollapse} which says that that result isn't known for regular collapse, just strong collapse.
%Mention that Whitehead's notion of collapse is weaker than this strong one (Barmak-Minian Prop 2.14). What is this exactly and how is it different from the strong notion? Why is the stronger one the right one to use? Is there something that can be done with the weaker one? Mention that this collapse uses the notation $\searrow$ so that's why we have two of those arrows, to indicate strong. We also have for the usual collapse that, if $K$ collapses to $K'$, then $|K'|$ is a deformation retract of $|K|$. Maybe for strong collapse you get strong deformation retract?

%%%%%%%%%%%%%%%%%%%%%%%%%%%%%%%%%%%%%%%%%%%%%%%%%%%%%%%%%%%%%%%%%%%%%%%%%%%%%%%%%%%%%%%%%%%%%%%%%%%%%%%%

\appendix

\section{The topology of simplicial complexes}\label{A:SimplicialComplexes}

%%%%%%%%%%%%%%%%%%%%%%%%%%%%%%%%%%%%%%%%%%%%%%%%%%%%%%%%%%%%%%%%%%%%%%%%%%%%%%%%%%%%%%%%%%%%%%%%%%%%%%%%

This section provides more detail about the topology of simplicial complexes that was briefly recalled in \refS{SimplicialComplexes}.
Some further references are  \cite{FP:SimpTop, K:CombAlgTop, M:Borsuk-Ulam, Spanier:Top}.

%%%%%%%%%%%%%%%%%%%%%%%%%%%%%%%%%%%%%%%%%%%%%%%%%%%%%%%%%%%%%%%%%%%%%%%%%%%%%%%%%%%%%%%%

\subsection{Basic definitions}\label{S:SimplCplx}

%%%%%%%%%%%%%%%%%%%%%%%%%%%%%%%%%%%%%%%%%%%%%%%%%%%%%%%%%%%%%%%%%%%%%%%%%%%%%%%%%%%%%%%%

We recall the definition of a simplicial complex for convenience.

\begin{defin}\label{D:SimplicialComplexAppendix}
An \textbf{(abstract) simplicial complex} $K = (V, \Delta)$ consists of a 
finite set $V$ whose elements are called \emph{vertices}  and a set $\Delta$ of subsets of $V$ called \emph{simplices}  satisfying
\begin{enumerate}
\item Elements of $V$ are in $\Delta$;
\item If $\sigma\in\Delta$ and $\tau\subset\sigma$, then $\tau\in\Delta$.
\end{enumerate}
\end{defin}

%We will sometimes denote the vertex set by $V(K)$. If $|V|=k+1$, we will also often label the elements of $V$ by $v_0, ..., v_{k}$.  The choice of enumeration of the elements will not be relevant.
%
%Since elements of $V$ are already listed in $\Delta$, we will often simply identify $K$ with $\Delta$ and will not write down $V$ explicitly.

As mentioned in \refS{SimplicialComplexes}, we will usually identify $K$ with $\Delta$. We will also sometimes denote the vertex set by $V(K)$. Also recall that the standard $n$-simplex consists of $n+1$ vertices and all possible nonempty subsets of the set of vertices.

\begin{example}\label{Ex:AbsSimplCplx}
An example of a simplicial complex on the vertex set  $V=\{v_0, v_1, ..., v_7\}$ is 
\begin{align*}
K= & \{
\{v_0, v_1, v_2, v_3\},
\{v_0, v_1, v_2\}, 
\{v_0, v_1, v_3\}, 
\{v_0, v_2, v_3\}, 
\{v_1, v_2, v_3\},
\{v_0, v_1\}, 
\{v_0, v_2\}, 
\{v_0, v_3\}, 
\{v_1, v_2\}, \\ & 
\{v_1, v_3\}, 
 \{v_2, v_3\}, 
\{v_3, v_4\}, 
\{v_3, v_5\}, 
\{v_4, v_5, v_6\}, 
\{v_4, v_5\}, 
\{v_4, v_6\}, 
\{v_5, v_6\}, 
\{v_4, v_7\},  \\
& \{v_0\}, \{v_1\}, \{v_2\}, \{v_3\}, \{v_4\}, \{v_5\}, \{v_6\}, \{v_7\}
\}
\end{align*}

A non-example is $V=\{v_0, v_1, v_2, v_3, v_4\}$ and a collection of subsets
$$
K=\{ 
\{v_1, v_2, v_3, v_4\},
\{v_1, v_3, v_4\},
\{v_0, v_2\},
\{v_1, v_2\},
\{v_1, v_4\},
\{v_2, v_3\},
\{v_3, v_4\},
\{v_0\}, \{v_1\}, \{v_2\}, \{v_3\}, \{v_4\}
\}
$$
The reason this is not a simplicial complex is that, for example, even though $\{v_1, v_2, v_3, v_4\}$ is in $K$, its subset $\{v_1,v_2, v_4\}$ is not.

\refEx{TopSimplCplx} shows a topological visualization of these examples.
\end{example}

\begin{example}\label{Ex:Standard}
The \emph{(standard) $n$-simplex} is a simplicial complex where $|V|=n+1$ and $K=\mathcal P_0(V)$, the power set of $V$ without the empty set.  Thus $K$ contains all possible nonempty subsets of $V$. 
\end{example}

\begin{example}
Recall that a \emph{graph} $\Gamma$ consists of a set $V$ of vertices and a set $E$ of edges defined as 
$$
E\subset\{\{v_1,v_2\}\colon v_1,v_2\in V, v_1\neq v_2\}.
$$ 
It is immediate from the definitions that a graph is precisely a 1-complex.
\end{example}

Here is some additional standard terminology we need:

\begin{defin}\label{D:ComplexStuff}\ 
\begin{itemize}
\setlength\itemsep{4pt}
\item If $\sigma$ is a simplex of $K$, then any subset $\tau$ of $\sigma$ is called a \emph{face} of $\sigma$.  The second condition in \refD{SimplicialComplex} thus says that all faces of a simplex are also simplices, i.e.~a simplicial complex is closed under taking faces.

\item
If $|\sigma|=n+1$, $n\geq 0$, then $\sigma$ is an \emph{$n$-simplex} (this is the same object as in Example \ref{Ex:Standard}, only now we think of it as an element of a simplicial complex). Vertices are thus 0-simplices.  
\item The \emph{dimension  of an $n$-simplex} is $n$.  The \emph{dimension of $K$} is 
$$\dim K = \max \{\dim \sigma \colon \sigma \in \Delta \}.$$ 
If the dimension of $K$ is $d$, we will say that $K$ is a \emph{$d$-complex}. 
\item A simplicial complex $L=(V',\Delta')$ is a \emph{subcomplex} of $K = (V, \Delta)$ if $V'\subset V$ and $\Delta'\subset\Delta$.
\item The \emph{$n$-skeleton} $K^{(n)}$ of $K$ is the collection of $n$-simplices of $K$ along with their faces. This is a subcomplex of $K$. There is a \emph{filtration} of $K$
$$
V=K^{(0)}\subset K^{(1)}\subset\cdots\subset K^{(d)}
$$
where $d$ is the highest simplex dimension occurring in $K$.
%\item Denote by $K(n)$ the set of $n$-simplices of $K$, i.e.~the set of elements of $\Delta$ of cardinality $n$.  This is not a subcomplex of $K$ because it does not contain the faces of the simplices in it.

\item The \emph{$f$-vector} of a non-empty $d$-dimensional complex $K$ is  
$$f=(f_0, f_1, ..., f_d)$$ where $f_n$, $0\leq n\leq d$, is the number of $n$-dimensional simplices of $K$.

\item A simplex that is not a face of another simplex is called a \emph{facet} or a \emph{maximal simplex}. A simplicial complex is determined by its maximal simplices.

\item A vertex $v$ is \emph{dominated by a vertex $w$} if every maximal simplex that contains $v$ also contains $w$.

\item A complex is \emph{minimal} if it has no dominated vertices.

\item The \emph{star of a simplex $\sigma$}, denoted by $\operatorname{st}(\sigma)$ is the subset of $K$ consisting of simplices that have $\sigma$ as a face, i.e.
$$\operatorname{st}(\sigma)=\{\tau\in K\colon \sigma\subset\tau\}.$$ 
This is not a subcomplex in general.  The special case we are most interested in is when $\sigma$ is a vertex.
%$v$
%, so 
%$$\st(v)=\{\tau\colon v\in\tau\}.$$

\item The \emph{closed star of $\sigma$}, or the \emph{closure of $\sigma$}, denoted by $\overline{\st}(\sigma)$ is the smallest subcomplex of $K$ containing $\st(\sigma)$. In other words,
$$\overline{\st}(\sigma)=\{\tau\in K\colon \sigma\cup\tau\ \text{is a simplex of $K$}\}.$$

\item The \emph{deletion} of a vertex $v$ is the subcomplex of $K$ obtained by removing from it $\st(v)$. In other words, this is the subcomplex spanned by all vertices of $K$ other than $v$. The deletion is denoted by $K\setminus v$.

\item The subcomplex of $\overline{\st}(\sigma)$ of simplices disjoint from $\sigma$ is called the \emph{link of $\sigma$}, denoted by $\lk(\sigma)$. In other words, $$\lk(\sigma)=\overline{\st}(\sigma)-\st(\sigma).$$

\end{itemize}
\end{defin}

%%%%%%%%%%%%%%%%%%%%%%%%%%%%%%%%%%%%%%%%%%%%%%%%%%%%%%%%%%%%%%%%%%%%%%%%%%%%%%%%%%%%%%%%

\subsection{Geometric realization}\label{S:Realization}

%%%%%%%%%%%%%%%%%%%%%%%%%%%%%%%%%%%%%%%%%%%%%%%%%%%%%%%%%%%%%%%%%%%%%%%%%%%%%%%%%%%%%%%%

For a simplicial complex $K$, there exists a topological space $|K|$ called the \emph{geometric realization} or the \emph{polyhedron} of $K$.  To construct it, first consider the \emph{standard (topological) $n$-simplex} $\Delta^n$ in $\R^{n+1}$ given as the convex hull of the standard basis vectors $e_1, ..., e_{n+1}$, i.e.~the set
$$
\left\{(t_1,...,t_{n+1})\in\R^{n+1} \colon   t_i\geq 0 \text{ for all $i$ and }\sum_{i=1}^{n+1}t_i=1 \right\}.
$$
This is topologized as a subspace of $\R^{n+1}$. Thus $\Delta^0$ is a point, $\Delta^1$ is a line segment, $\Delta^2$ is a (filled-in) triangle, $\Delta^3$ is a (solid) tetrahedron, etc. Faces of $\Delta^n$ are obtained by setting some of the $t_i$ to zero.

The realization $|K|$ is then constructed by taking a copy of $\Delta^n$ for each $n$-simplex of $K$, $n\geq 0$ and gluing them along common faces. Thus $|K|$ is the quotient of the disjoint union $\amalg_\alpha \Delta^\alpha$ where $\alpha$ runs over all simplices of $K$. There are various ways to choose the copies of the standard simplices and glue them, but it turns out that they produce homeomorphic spaces (see \refP{SimplIsoHomeo} below). It is also a well-known result that an abstract $d$-complex $|K|$ has a geometric realization in $\R^{2d+1}$ (see, for example \cite[Section 1.6]{M:Borsuk-Ulam}).

One can also think of the geometric realization as taking $|V|$ affinely independent points (so they do not lie on an $(|V|-1)$-dimensional hyperplane) in $\R^N$ (for $N$ sufficiently large) and filling in a topological simplex (convex hull) for each subset of those points that appears as a subset in $\Delta$.  This means that, instead of taking the full $|V|$-simplex, we take only those of its faces as prescribed by $\Delta$. The result is topologized as a subspace of $\R^N$. We will omit the details here and will illustrate the construction with an example.  

We will also, by abuse of notation, label the affinely independent points in $\R^N$ the same way as the vertices in $V$.

\begin{example}\label{Ex:TopSimplCplx}
The geometric realization of the simplicial complex from Example \ref{Ex:AbsSimplCplx} is given below in the left picture of Figure \ref{F:TopSimplicialComplex}. The right picture provides a visualization of how the second part of that example fails to give a simplicial complex. Note that by abuse of notation we are labeling the vertices of the realization the same way as the elements of the vertex set of the abstract simplicial complex.
\begin{figure}[h]
\centering
\includegraphics[width=2in]{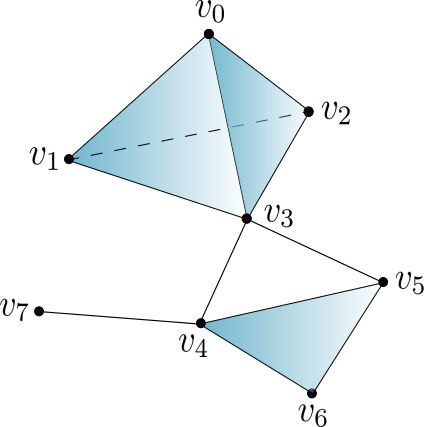}
\ \ \ \ \ \ \ 
\includegraphics[width=1.7in]{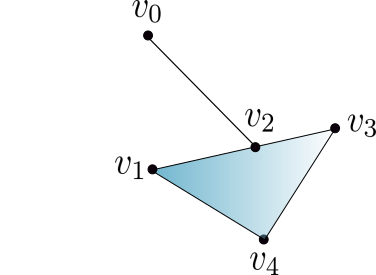}
\caption{An example and a non-example of a simplicial complex.}
\label{F:TopSimplicialComplex}
\end{figure}

\end{example}

In terms of geometric realizations, one way to think about the star of a vertex (\refD{ComplexStuff}) is as the analog of an open neighborhood of that vertex, closed star as its closure, and link as the boundary of the star. For example, the star of the vertex $v_3$ in the simplicial complex of Example \ref{Ex:TopSimplCplx} consists of the tetrahedron but with the interior of the triangle with vertices $v_0, v_1, v_2$ removed as well as the two edges emanating from $v_3$ and ending at $v_4$ and $v_5$. Abstractly, even though $\st(v_3)$ contains the sets $\{v_0, v_1, v_2, v_3\}$, $\{v_3, v_4\}$, and $\{v_3, v_5\}$, it does not contain all the subsets of those sets ($\{v_0, v_1, v_2\}$, $\{v_3\}$, and $\{v_4\}$) and that is why it is not a subcomplex.

\subsection{Operations on simplicial complexes}\label{S:ComplexOperations}

%%%%%%%%%%%%%%%%%%%%%%%%%%%%%%%%%%%%%%%%%%%%%%%%%%%%%%%%%%%%%%%%%%%%%%%%%%%%%%%%%%%%%%%%

There are several useful constructions one can perform on simplicial complexes. These correspond to operations on political structures in \refS{PoliticsMods}.

%The most basic construction is the \emph{product} of simplicial complexes $K$ and $L$, denoted as usual by $K\times L$. We will not define this precisely here since we will not need it in a serious way.  Suffice it to say that, to define this product, one first has to decide on a coherent way of turning the set product of two simplices $\sigma$ and $\tau$ into a simplicial complex (i.e.~\emph{triangulating} it). The vertices of the product complex are pairs of vertices, one from $\sigma$ and one from $\tau$ and simplices are determined by using an ordering on the vertices of $\sigma$ and $\tau$. This procedure produces a simplicial complex $K\times L$ with the important property that it commutes with realization.  Thus
%$$
%|K\times L|\cong |K|\times |L|
%$$
%where by the right side we mean the usual cartesian product of topological spaces with the product topology.
%(Can cite beginning of Chapter III of Ferrario and Picinnini.)
%

\begin{defin}\label{D:Wedge}
Let $K$ and $L$ be simplicial complexes with distinguished 0-simplices (vertices) $v_0$ and $w_0$.  Define the \emph{wedge} of $K$ and $L$, denoted by $K\vee L$, to be the simplicial complex obtained as the union of $K$ and $L$, except that $v_0$ and $w_0$ are relabeled and denoted the same way. This common relabeled point is called the \emph{wedge point}.
\end{defin}

%So  if $\{v_0, x,y  \}$ is a simplex in $K$ and $\{w_0,a,b,c  \}$ is a simplex in $L$, $K\vee L$ would contain simplices $\{v, x,y  \}$ and $\{v,a,b,c  \}$ (using $v$ as the new common label for $v_0$ and $w_0$).

On the level of realizations, $|K|\vee |L|$ is obtained as a quotient space
$$
|K|\vee |L| = (|K|\sqcup |L|)/(v_0\sim w_0),
$$
i.e.~as the disjoint union of $|K|$ and  $|L|$ with the two distinguished points identified. It is a standard result that the wedge commutes with realization, i.e.~there is a homeomorphism

$$
|K\vee L|\cong |K|\vee |L|.
$$

\begin{example}
The simplicial complex in Example \ref{Ex:TopSimplCplx} can be thought of as a wedge of two complexes with $v_3$ as the wedge point.
\end{example}

The wedge is a special case of the following construction.

\begin{defin}\label{D:Pushout}
Suppose $K$ and $L$ are simplicial complexes and let $S$ be a subcomplex of both (i.e.~$K$ and $L$ contain isomorphic subcomplexes, both denoted by $S$).  Let $f\colon S\hookrightarrow K$ and $g\colon S\hookrightarrow L$ be the inclusion maps of the subcomplex.  Define the \emph{pushout} or \emph{colimit} of the diagram $K\stackrel{f}\hookleftarrow  S\stackrel{g}\hookrightarrow L$, denoted by $K\amalg_S L$ (or $\operatorname{colim}(K\stackrel{f}\hookleftarrow  S\stackrel{g}\hookrightarrow L)$),  to be
$$
%\operatorname{colim}(K\stackrel{f}{\hookleftarrow}  S\stackrel{g}{\hookrightarrow} L)
K\amalg_S L = 
(K\amalg L)/\sim
$$
where $\sim$ is the equivalence relation generated by $f(s)\sim g(s)$ for $s\in S$.
\end{defin}

Thus the pushout of $|K|$ and $|L|$ is the disjoint union of these two spaces but glued along the common subspace which is the realization of the common subcomplexes.  When $S=v$, a single vertex, we get precisely the wedge:
$$
%\operatorname{colim}(K\stackrel{f}{\hookleftarrow}  v_0\stackrel{g}{\hookrightarrow} L) 
K\amalg_{v} L
\cong 
K\vee L.
$$

One special case of the pushout is the procedure of \emph{attaching $n$-cells}. This is how the realization $|K|$ can be defined inductively by starting with affine points in $\R^N$ and then building up the realization by attaching cells of higher and higher dimensions using pushouts.

\begin{defin}\label{D:Join}
For $K$ and $L$ simplicial complexes, define the \emph{join} of $K$ and $L$, denoted by $K\ast L$, to be the simplicial complex whose vertices are the union of the (distinct) vertices in $K$ and $L$ and whose simplices are those of $K$, $L$, and unions of simplices in $K$ and $L$, i.e.
$$
K\ast L = K\sqcup L\sqcup \{\sigma\cup \tau\colon \sigma\in K, \tau\in L\}.
$$
\end{defin}

%\begin{example}
%Need example, just on level of abstract complexes.
%\end{example}

It is also not hard to see that the realization commutes with join. In other words, there is a homeomorphism
$$
|K\ast L|\cong |K|\ast |L|,
$$
where $|K|\ast |L|$ is the join of topological spaces defined as
$$
|K|\ast |L| = (|K|\times |L|\times I) /\sim
$$
with $I$ the unit interval and $\sim$ the equivalence relations generated by 
$$
(x',y,0)\sim (x,y,0) \text{ and } (x,y,1)\sim (x,y',1)  \text{ for all } x,x'\in |K| \text{ and }y,y'\in |L|.
$$ 
Intuitively $|K|\ast |L|$ is the space of all line segments from $|K|$ to $|L|$.

%\begin{example}
%Need example (pictures).
%\end{example}

\begin{example}\label{E:ConeSuspension} Here we provide some important special cases of the join.

\begin{itemize}
\setlength\itemsep{4pt}
\item If $L=\{c\}$ is a single vertex, then $K\ast \{c\}$ is called the \emph{cone on $K$}, denoted by $CK$. Vertex $c$ is the \emph{cone vertex} or the \emph{cone point}.  Geometrically, the cone $C|K|$ on $|K|$ is the quotient $(|K|\times I)/(|K|\times \{1\})$. An example is given in Figure \ref{F:Cone}.

\vspace{4pt}

\begin{itemize}
\item[$\circ$] One example of a cone is 
$$
C\Delta^n = \Delta^{n+1}. $$ 
Thus $C|K|\cong |CK|$ can be thought of as increasing the dimension of each simplex of $|K|$ by joining all points of $|K|$ to a common disjoint point $c$ by line segments.  
\end{itemize}

\item If $L=\{c, d\}$ consists of two vertices, then $K\ast \{c,d\}=\Sigma K$ is called the \emph{suspension of $K$}. On the level of realizations, the suspension is the quotient $(|K|\times I)/(|K|\times \{0\}, |K|\times \{1\})$. 

\vspace{4pt}

\begin{itemize}
\item[$\circ$] The suspension of a simplex, $\Sigma|\Delta^n|\cong|\Sigma\Delta^n|$, is obtained by taking two copies of $|\Delta^{n+1}|$ and gluing them along a common $\Delta^n$-face.  The suspension of $|K|$ can then be thought of as taking two copies of the cone on $|K|$ and gluing them along $|K|$.  In other words,  $\Sigma|K|$ is the pushout of two copies of $C|K|$ along the common subspace $|K|$.
\vspace{4pt}
\item[$\circ$] Suppose $K$ contains a subcomplex $L$ whose realization is homeomorphic to the $n$-dimensional sphere $S^n$.  For example, a closed path of 1-simplices has $S^1$ as its realization, a hollow tetrahedron is $S^2$, etc. Then $|\Sigma L|$ is homeomorphic $S^{n+1}$. (This leads to the fact that a homology class suspends to the class of one higher dimension and that, in turn, explains why the homology of a complex and its suspension are closely related; see part (3) of \refP{HomologyOfOperations}.)
\end{itemize}

\end{itemize}

\end{example}

%\begin{example}\label{Ex:Cone}
%Figure \ref{F:Cone} below gives an example of a simplicial complex $|K|$ and a cone on $|K|$.  
\begin{figure}[h]
\centering
\includegraphics[width=2in]{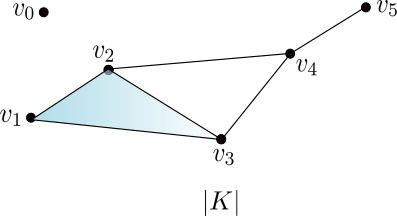}
\ \ \ \ \ \ \ 
\includegraphics[width=2in]{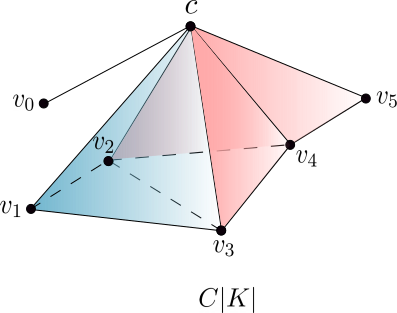}
\caption{An example of a cone on a simplicial complex. In the right picture, the tetrahedron with vertices $c, v_1, v_2, v_3$ is filled in, while the one with vertices $c, v_2, v_3, v_4$ is hollow. The triangle $v_2, v_3, v_4$ is also not filled in.}
\label{F:Cone}
\end{figure}

%\end{example}

Recall from \refD{ComplexStuff} that a vertex $v$ is dominated by a vertex $w$ if every maximal simplex that contains $v$ also contains $w$. Recall also that the deletion of a vertex $v$, $K\setminus v$, is obtained from $K$ by deleting $\st(v)$, the set of all simplices containing $v$. Essentially by definition, the closed star $\overline{\st}(v)$ is a cone on $v$.  We also have the following useful connection between cones and deletions (see \cite[Definition 2.1 and Remark 2.2]{BM:StrongHtopy}).

\begin{prop}\label{P:DominatedLink}
Vertex $v$ is dominated by a vertex $w$ if and only if the link of $v$ is a cone on $w$.
\end{prop}

%
%\begin{defin}{Suspension}
%\newline Let $\Delta$ be a simplicial complex. Adding twl vertices u and v $\notin$ V, the vertex set of the simplicial complex. The \textbf{suspension} of a simplicial complex is $\sum$($\Delta$) = $\{$ $\sigma \cup \{ u \}$ and $\sigma \cup \{ v \}$, where $\sigma \in \Delta\}$
%\end{defin}

%
%\begin{defin}{Join}
%\newline Let $\Delta_1$ and $\Delta_2$ be simplicial complexes. The \textbf{join} of  $\Delta_1$ and $\Delta_2$ denoted as  $\Delta_1$*$\Delta_2$ is the set of vertices of $\Delta_1$*$\Delta_2$ is V( $\Delta_1$) $\cup$ V( $\Delta_2$). The set of simplices is then described by $\Delta_1$*$\Delta_2$ = $\{ \sigma \subseteq$ V( $\Delta_1$)  $\cup$ V( $\Delta_2$) $| \sigma \cap$ V( $\Delta_1$) $\in \Delta$ and  $\sigma \cap V( \Delta_2) \in \Delta \}$
%\end{defin}
%The join of a simplicial complex can be seen as a way to produce new simplicial complexes from the already existing ones. 

%%%%%%%%%%%%%%%%%%%%%%%%%%%%%%%%%%%%%%%%%%%%%%%%%%%%%%%%%%%%%%%%%%%%%%%%%%%%%%%%%%%%%%%%

%\section{Simplicial maps and equivalences}\label{S:SimplMapEquiv}

%%%%%%%%%%%%%%%%%%%%%%%%%%%%%%%%%%%%%%%%%%%%%%%%%%%%%%%%%%%%%%%%%%%%%%%%%%%%%%%%%%%%%%%%

%%%%%%%%%%%%%%%%%%%%%%%%%%%%%%%%%%%%%%%%%%%%%%%%%%%%%%%%%%%%%%%%%%%%%%%%%%%%%%%%%%%%%%%%

\subsection{Simplicial maps}\label{S:SimplMap}

%%%%%%%%%%%%%%%%%%%%%%%%%%%%%%%%%%%%%%%%%%%%%%%%%%%%%%%%%%%%%%%%%%%%%%%%%%%%%%%%%%%%%%%%

%Simplicial maps are a way to compare simplicial complexes, just as continuous maps are a way to compare topological spaces. Of special importance to us is the idea of contiguous simplicial maps and the related notion of a strong collapse.

\begin{defin}\label{D:SimplMap}
Let $K$ and $L$ be simplicial complexes.  A \emph{simplicial map} $\phi\colon K\to L$ from  $K$ to $L$ is a function that sends vertices to vertices and simplices to simplices.  In other words, $\phi(V(K))\subset V(L)$ and, for each simplex $\sigma$ in $K$, $\phi(\sigma)$ is a simplex in $L$.
\end{defin}

A bijective simplicial map whose inverse is also simplicial is called an \emph{isomorphism}. If an isomorphism exists between $K$ and $L$, we write $K\cong L$ and consider the two simplicial complexes as the same (up to relabeling of vertices).

A simplicial map $\phi\colon K\to L$ induces a map of topological spaces
$$
|\phi|\colon |K|\longrightarrow |L|
$$
defined as the linear extension of the map on vertices given by $\phi$.  Namely, suppose $v_0, ..., v_k$ are the affinely independent points in $\R^N$ (remember that we are labeling these the same way as the vertices in $V$). Each point $v\in |K|$ is thus given as $v=\sum_{i=0}^k t_iv_i$, $t_i\geq 0$, $\sum_{i=0}^k t_i=1$. Then define $|\phi|$ as 
$$
|\phi|(v)=\sum_{i=1}^k t_i\phi(v_i).
$$

Thus a simplicial map is determined by what it does on vertices.

%
%
%\begin{example}\label{Ex:SimplMap}
%Need example, on abstract level and corresponding realization level.
%\end{example}

The following is a standard result (see, for example, \cite[Theorem II.2.7]{FP:SimpTop}).

\begin{prop}\label{P:SimplIsoHomeo}
If $\phi\colon K\to L$ is a simplicial map, then $|\phi|\colon |K|\to |L|$ is continuous.  If $\phi$ is injective, so is $|\phi|$.  If  $\phi$ is an isomorphism, then $|\phi|$ is a homeomorphism.
\end{prop}

The collection of simplicial complexes forms a category $\operatorname{SCom}$.  The morphisms are the simplicial maps. It is not hard to see that the axioms of a category are satisfied; the most relevant one being that the composition of simplicial maps is a simplicial map.

In category theory language, geometric realization is a functor
$$
|\cdot|\colon \operatorname{SCom}\longrightarrow \operatorname{Top}
$$
from the category of simplicial complexes to the category of topological spaces (with, as usual, continuous maps as morphisms).

From now on, we will not make a distinction between the abstract simplicial complex $K$ and its realization $|K|$ but will simply use $K$ for both.  A simplicial map $\phi\colon K\to L$ will thus simultaneously be the set map between abstract simplicial complexes and a map of the corresponding topological simplicial complexes, i.e.~their geometric realizations.

%%%%%%%%%%%%%%%%%%%%%%%%%%%%%%%%%%%%%%%%%%%%%%%%%%%%%%%%%%%%%%%%%%%%%%%%%%%%%%%%%%%%%%%%

\subsection{Collapses and strong equivalences}\label{S:Collapse}

%%%%%%%%%%%%%%%%%%%%%%%%%%%%%%%%%%%%%%%%%%%%%%%%%%%%%%%%%%%%%%%%%%%%%%%%%%%%%%%%%%%%%%%%

The usual notion of ``sameness'' in algebraic topology is that of homotopy equivalence: Spaces $X$ and $Y$ are \emph{homotopy equivalent} if there are maps $f\colon X\to Y$ and $g\colon Y\to X$ such that the compositions $f\circ g$ and $g\circ f$ are homotopic to identity maps on $Y$ and $X$, respectively.

%
%\begin{defin}\label{D:Heq}\ 
%\begin{itemize}
%\setlength\itemsep{4pt}
%\item Maps $\alpha, \beta\colon X\to Y$ are \emph{homotopic} if there is a one-parameter family of maps $F\colon X\times I\to Y$ such that the restrictions to the ends of the interval $I=[0,1]$ are $\alpha$ and $\beta$.
%\item Spaces $X$ and $Y$ are \emph{homotopy equivalent}, or they have the same \emph{homotopy type}, denoted by $X\simeq Y$, if there are maps $f\colon X\to Y$ and $g\colon Y\to X$ such that the compositions $f\circ g$ and $g\circ f$ are homotopic to identity maps on $Y$ and $X$, respectively.  In this case, $f$ and $g$ are called \emph{homotopy equivalences}.
%\item A space is said to be \emph{contractible} if it is homotopy equivalent to a one-point space.
%\end{itemize}
%\end{defin}

The heuristic is that homotopy equivalent spaces can be ``deformed'' into one another. This notion in particular applies to realizations of simplicial complexes -- if $|K|$ and $|L|$ are homotopy equivalent, then one is  deformable into the other.  However, this deformation may not always go through simplicial complexes, i.e.~there might be some $t\in [0,1]$ during the homotopy such that the image of one of the compositions is not a simplicial complex. In other words, a homotopy equivalence may not always arise in a natural way (via geometric realizations and maps between them induced by simplicial maps) from a procedure in the category of abstract simplicial complexes.  What is thus needed is a construction on simplicial complexes that induces a homotopy equivalence on realizations, and it does so in the ``best possible'' way.

%With this in mind, we have

\begin{defin}\label{D:Contiguity}
Suppose $\phi, \psi\colon K\to L$ are simplicial maps.  Then $\phi$ and $\psi$ are \emph{contiguous} if, for every simplex $\sigma\in K$, $\phi(\sigma)\cup \psi(\sigma)$ is a simplex in $L$. 
\end{defin}

It follows from the definition that $\phi$ and $\psi$ are contiguous if, given a simplex $\sigma$ of $K$ spanned by vertices $v_0, ..., v_n$, then $\phi(v_0), ..., \phi(v_n), \psi(v_0), ..., \psi(v_n)$ span a simplex of $L$.

\begin{example}
Let 
\begin{align*}
K & =\{\{v_0\}, \{v_1\}, \{v_2\}, \{v_0,v_1\}, \{v_1,v_2\} \}\\ 
L & =\{\{w_0\}, \{w_1\}, \{w_2\}, \{w_0, w_1\}, \{w_0, w_2\}, \{w_1, w_2\}, \{w_0, w_1, w_2\}\}.
\end{align*} 
Let $\phi, \psi\colon K\to L$ be simplicial maps defined by sending vertices to correponding vertices, along with 
$$
\phi(\{v_0,v_1\})=\{w_0,w_1, w_2\},\  \phi(\{v_1,v_2\})=\{w_1, w_2\},\ 
  \psi(\{v_0,v_1\})=\{w_0,w_1\},\       \psi(\{v_1,v_2\})=\{w_0, w_1, w_2\}.
$$
Then $\phi$ and $\psi$ are contiguous simplicial maps since, for each simplex in $K$, the unions of their images under the two maps are either the vertices $\{w_i\}$ or the simplex $\{w_0, w_1, w_2\}$.
\end{example}

\begin{example}
As a non-example, consider $K  =L=\{\{v_0\}, \{v_1\}, \{v_2\}, \{v_0,v_1\}, \{v_0,v_2\}, \{v_1,v_2\} \}$. Let $\phi$ be the identity map and let $\psi$ be the constant map sending everything to $\{v_0\}$. Then these maps are not contiguous since $\phi(\{v_1,v_2\})\cup \psi(\{v_1,v_2\}) = \{v_0, v_1,v_2\}$, but this is not a simplex in $K$.
\end{example}

If there is a sequence of contiguous simplicial maps connecting $\phi$ and $\psi$, we say that these maps are in the same \emph{contiguity class} and write $\phi\sim_c\psi$. 

%Compare the following with the definition of homotopy equivalent spaces above.

\begin{defin}\label{D:StrongEquivalence} \ 
\begin{itemize}
\setlength\itemsep{4pt}
\item 
Complexes $K$ and $L$ are \emph{strongly equivalent} if there are simplicial maps $\phi\colon K\to L$ and $\psi\colon L\to K$ such that the compositions satisfy $\phi\circ \psi\sim_c \operatorname{Id}_{L}$ and $\psi\circ \phi\sim_c \operatorname{Id}_{K}$. In this case $\phi$ and $\psi$ are called \emph{strong equivalences}. We write $K\sim_c L$.
%\item If $K\sim_c L$, we say that $K$ and $L$ have the same \emph{strong homotopy type}.
\item Complex $K$ is \emph{strongly contractible} if it is strongly equivalent to the  single-vertex complex, i.e.~if the identity map on $K$ is contiguous to the constant map sending $K$ to one of its vertices.
\end{itemize}
% and $K$ and $L$ are said to have the same \emph{strong homotopy type}.
\end{defin}

%A complex $K$ is \emph{strongly contractible} if it is strongly equivalent to the single-vertex complex.

The following is easy to see (using straight-line homotopy).  For details, see for example \cite[Corollary A.1.3]{B:FiniteTop}.
\begin{prop}\label{P:ContiguousHomotopic}
If $\phi, \psi\colon K\to L$ are contiguous, then $|\phi|, |\psi|\colon |K|\to |L|$ are homotopic.  
\end{prop}

 %In particular, a strongly contractible complex $K$ has a contractible realization $|K|$.  
 
% The homotopy in this result is an fact a strong deformation retract.

The converse of \refP{ContiguousHomotopic} is not true. The standard examples are ``Bing's house with two rooms'' and the ``dunce hat'' which are contractible but can be exhibited as realizations of simplicial complexes that are not strongly equivalent to one-point complexes.

\begin{rem}\label{R:BarycenterContiguity}
There is a partial converse to \refP{ContiguousHomotopic} which says that if $f$ and $g$ are homotopic maps from $|K|$ to $|L|$, then there exist simplicial maps $\phi$ and $\psi$ from $K$ to $L$ that are not necessarily contiguous, but if $K$ is sufficiently barycentrically divided, then they are (see, for example, \cite[Theorem 10.21]{McCleary:ContDim}).  This is the sense in which contiguity is the best analog of homotopy in the world of abstract simplicial complexes.
\end{rem}

Strong equivalences can be characterized by sequences of special kinds of moves. Details can be found in \cite[Chapter 5]{B:FiniteTop} and \cite{BM:StrongHtopy}.  To explain, recall the terminology from \refD{ComplexStuff}. 

\begin{defin}\label{D:ElemStrongCollapse}
The deletion of a dominated vertex $v$ from the simplicial complex $K$ is called an \emph{elementary strong collapse}, denoted by $K\, \esc\, K\setminus v$. 
\end{defin}

According to \refP{DominatedLink}, an elementary strong collapse to the subcomplex obtained by deleting (the star of) $v$ exists if the link of $v$ is a cone.

\begin{defin}\label{D:StrongCollapse} \ 
\begin{itemize}
\setlength\itemsep{4pt}
\item A sequence of elementary strong collapses is called a \emph{strong collapse}.
\item The inverse of a strong collapse is a \emph{strong expansion}. 
\item If complexes $K$ and $L$ are connected by a sequence of strong collapses and expansions, then they have the same \emph{strong homotopy type}.
\item A complex $K$ is \emph{strongly collapsible} if it has the strong homotopy type of a single-vertex complex.
\end{itemize}
\end{defin}

\begin{example}\label{Ex:StrongCollapse}
Figure \ref{F:StrongCollapse}  gives an example of a strong collapse. Instead of writing this down on the level of abstract simplicial complexes, we provide the example on the level of geometric realizations to help with the visualization.

\begin{figure}[h]
\centering
\includegraphics[height=1.5in]{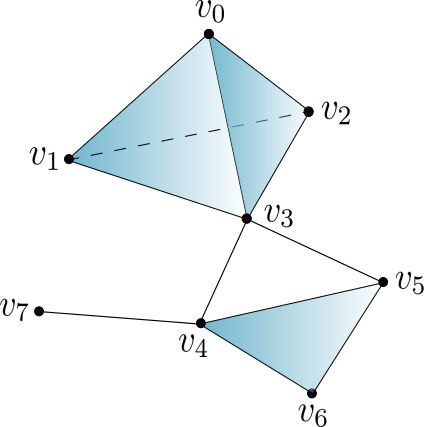}
\ \ \ \ \ \ \ 
\includegraphics[height=1.2in]{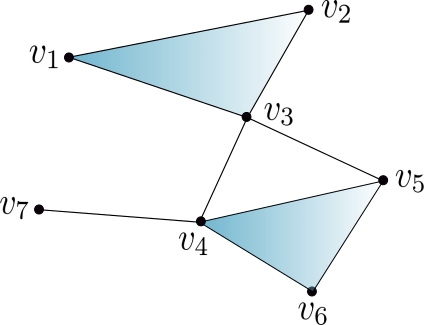}
\ \ \ \ \ \ \ 
\includegraphics[height=1.2in]{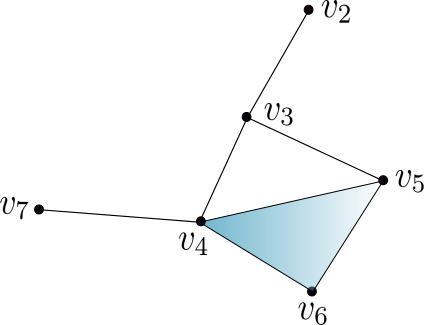}

\vspace{12pt}

\includegraphics[height=0.9in]{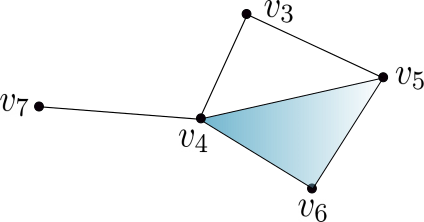}
\ \ \ \ \ \ \ 
\raisebox{0.8cm}{
\includegraphics[height=0.6in]{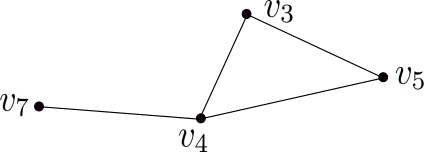}
\ \ \ \ \ \ \ 
\includegraphics[height=0.6in]{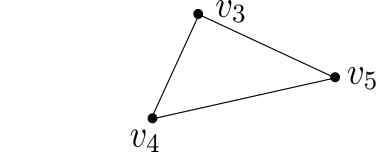}}
\caption{An example of a strong collapse. The vertices deleted, in order, are $v_0$, $v_1$, $v_2$, $v_6$, and $v_7$.}
\label{F:StrongCollapse}
\end{figure}

\end{example}

Recall that a complex without any dominated vertices is called minimal.  Thus a minimal complex is one which does not admit any elementary strong collapses. The minimal complex obtained by the reduction of a complex $K$ to its minimal subcomplex via a strong collapse is called the \emph{core of $K$}, denoted by $K^c$.  For example, the triangle with vertices $v_3$, $v_4$, and $v_5$ in Figure \ref{F:StrongCollapse} is the core of all the simplicial complexes in that picture.

\begin{thm}[\cite{B:FiniteTop},  Theorem 5.1.10]\label{T:StrongCores}
The core of a simplicial complex is unique up to isomorphism. Two complexes have the same strong homotopy type if and only if they have isomorphic cores.
\end{thm}

One observation is that the previous theorem guarantees that the order in which elementary strong collapses are performed is irrelevant since each sequence of such moves must yield the same core.

Then we finally have the following consequence. Recall \refD{StrongEquivalence}.

\begin{cor}[\cite{B:FiniteTop},  Corollary 5.1.11]\label{C:Contiguous=StrongCollapse}
Complexes $K$ and $L$ are strongly equivalent if and only if they have the same strong homotopy type.
\end{cor}

Thus the notion of equivalence via contiguous maps (strong equivalence) is the same as that of equivalence via strong collapses (strong homotopy type). 
In particular, 
 a complex $K$ is strongly equivalent to a vertex if and only if it is strongly collapsible to a vertex. 
 %In this situation, we will say that $K$ is \emph{strongly contractible}.  
 
 Since the two notions are the same, they have the same effect upon applying the realization functor. Namely, if there is a strong collapse/strong equivalence from $K$ to a subcomplex $L$, then we know by \refP{ContiguousHomotopic} that there is homotopy equivalence between $|K|$ and the subspace $|L|$. However, this homotopy equivalence is in fact a strong deformation retraction. In particular, if $K$ is strongly contractible, then $|K|$ strong deformation retracts to a point. In the category of spaces, one also calls this strong contractibility, so the terminology from the categories of simplicial complexes and topological spaces also agrees. A strongly contractible space is of course contractible in the usual sense, i.e.~homotopy equivalent to the one-point space.

%\begin{example}
%Maybe have an example with a sequence of elementary collapses that reduces to a point. Andrea has some in her notes from 6/17.
%\end{example}

We close this section with a few relevant results. Suppose $K$ and $L$ are complexes and $S$ is a common subcomplex. Recall the notions of the pushout $K\amalg_S L$, wedge $K\vee L$, join $K\ast L$, cone $CK$, and suspension $\Sigma K$ from \refS{ComplexOperations}. Also recall that the core of a simplex $K$ is denoted by $K^c$.

\begin{prop}\label{P:PushoutCollapse}\ 
\begin{enumerate}
\item $(K\amalg_S L)^c = S \ \ \Longleftrightarrow \ \ S= K^c \ \text{and}\  S= L^c$.
\item $(K\amalg_S L)^c = K^c\ (\text{resp.~$L^c$})\ \  \Longleftrightarrow\ \  L\, \esc\, S$\ (\text{resp.~$K\, \esc\, S$}).
\end{enumerate} 
\end{prop}

The two statements intersect in the special case when $S=L^c$ (resp.~$S=K^c$).

\begin{proof} For the forward direction of $(1)$, if $(K\amalg_S L)^c = S$, then one can get from $S$ to $K\amalg_S L$ via a sequence of strong expansions. Each of those expansions can only produce vertices in $K$ or $L$, but not both since $S$ already contains all the vertices that are common to $K$ and $L$. Those strong expansions can then be performed in a way that they first generate all the vertices of $K$ that are not in $S$ and then, again necessarily starting from $S$, all the vertices of $L$ that are not in $S$. Since $S$ is already minimal, this shows that $S$ is the core of both $K$ and $L$.

For the other direction, if $S= K^c$ and $S= L^c$, then any vertex $v$ that is dominated by a vertex $w$ in $K$ (resp.~$L$) is still dominated by the same vertex in $K\amalg_S L$. Elementary collapses can then be performed independently in the copies of $K$ and $L$ in $K\amalg_S L$, and we hence have $K\amalg_S L \, \esc\,  K^c \amalg_S L^c=S$.

For the forward direction of $(2)$, suppose $L$ is not strongly collapsible to $S$. This means that there exists a vertex $v\notin S$ which is not dominated by any vertex of $S$. This also then means that $v$ is not dominated by any vertex of $K$ in $K\amalg_S L$. Thus the core of $K\amalg_S L$ must contain a vertex in $L$ that is not in $K$, and $(K\amalg_S L)^c$ cannot be $K^c$.

For the other direction, since $L\, \esc\, S$, there is a sequence of elementary collapses that reduces $K\amalg_S L$ to $K$, i.e.~$K\amalg_S L\, \esc\, K$. Following these collapses by collapses in $K$ that reduce it to its core gives $K\amalg_S L\, \esc\, K^c$.
\end{proof}

A special case of \refP{PushoutCollapse} is when $S$ is a single vertex, so $K\amalg_S L=K\vee L$. We obtain the following consequence which will be used in \refS{Delegations}.

\begin{cor}\label{C:WedgeCollapse}
The wedge $K\vee L$ is strongly collapsible if and only if $K$ and $L$ are strongly collapsible. We have $(K\vee L)^c=K^c$ (resp.~$L^c$) if and only if $L$ (resp.~$K$) is strongly collapsible. 
\end{cor}

\begin{prop}[\cite{BM:StrongHtopy}, Proposition 2.15]\label{P:JoinCollapse}
The join of two complexes $K\ast L$ is strongly collapsible if and only if $K$ or $L$ is strongly collapsible.
\end{prop}

Since a single-vertex complex is strongly collapsible but the two-vertex complex is not, we immediately get the following consequence.

\begin{cor}\label{C:ConeSuspensionStrong}
For any complex $K$, the cone  $CK$ on $K$ is strongly collapsible.  The suspension $\Sigma K$ of $K$ is strongly collapsible if and only if $K$ is. 
\end{cor}

One can generalize this corollary to the join with any 0-complex with the number of vertices greater than two, in which case the join is again strongly collapsible if and only if $K$ is.

\subsection{Homology of simplicial complexes}\label{S:Homology}

%%%%%%%%%%%%%%%%%%%%%%%%%%%%%%%%%%%%%%%%%%%%%%%%%%%%%%%%%%%%%%%%%%%%%%%%%%%%%%%%%%%%%%%%

Here we give a brief overview of homology. For more details about homology of spaces, see \cite{Hatcher}. For homology of simplicial complexes, see \cite{FP:SimpTop, K:CombAlgTop}.

Given an $n$-simplex $\sigma$, there are $(n+1)!$ ways to order its vertices.  We say that the orientations are \emph{equivalent} if they differ by an even permutation.   For $n>0$, there are thus two equivalence classes of orientations. A choice of one of the classes will be called an \emph{orientation} of $\sigma$.  If  $\sigma=\{v_0, ..., v_{n}\}$, then $\sigma$ with a choice of an orientation will be denoted by
$
\sigma=[v_0, ..., v_{n}].
$

If $\sigma=[v_0, ..., v_{n}]$ has a natural orientation determined by the ordering of its vertices according to their label and if $\tau=\{v_0, ...,\hat{v}_i,..., v_{n}\}=\{v_0, ...,v_{i-1}, v_{i+1},... v_{n}\}$ is the face of $\sigma$ obtained by removing the vertex $v_i$, then a natural orientation on $\tau$ is $(-1)^{i}[v_0, ...,v_{i-1}, v_{i+1},..., v_{n}]$. So for example, if $\sigma=[v_0, v_1, v_2]$ is a 2-simplex, then the edges are oriented as $[v_0,v_1]$, $[v_1,v_2]$, and $[v_2,v_0]=-[v_0,v_2]$.

A simplicial complex $K$ is \emph{oriented} if all its simplices are oriented. One way to orient a simplicial complex is to order its vertices and let this ordering induce an orientation on all the simplices of $K$.

\begin{defin}\label{D:n-chains}
Let $K$ be an oriented simplicial complex. For $n\geq 0$, define the \emph{group of $n$-chains of $K$}, denoted by $\Ch_n(K)$, to be the free abelian group generated by the $n$-simplices of $K$. In other words, if $\{\sigma^i\}_{i\in I}$ is the finite set of oriented $n$-simplices of $K$, then an \emph{$n$-chain} is a formal linear sum 
$$
\sum_{i\in I} c_i\sigma^i, \ \ c_i\in\Z,
$$
with the addition given by
$$
\sum_{i\in I} c_i\sigma^i+\sum_{i\in I} d_i\sigma^i = \sum_{i\in I} (c_i+d_i)\sigma^i.
$$
For $n=-1$, define $\Ch_{-1}(K)=\Z$ and for $n<-1$, define $\Ch_n(K)=0$.
\end{defin}

For $n\neq 0$, one has a homomorphism, called the \emph{boundary operator}, 
\begin{equation}\label{E:Boundary}
\partial_n\colon \Ch_n(K)\longrightarrow \Ch_{n-1}(K)
\end{equation}
defined as follows: On an oriented simplex $[v_0,...,v_n]$, let 
$$
\partial_n([v_0,...,v_n])=\sum_{i=0}^n (-1)^i [v_0, ...,v_{i-1}, v_{i+1},..., v_{n}].
$$
Then to define \eqref{E:Boundary}, extend this map linearly to chains. Also define $\partial_{0}$ as the sum of the coefficients of a 0-chain. If $K$ is a $d$-dimensional complex, we thus get a sequence of groups and homomorphisms
$$
\cdots \longrightarrow 0 \stackrel{0}{\longrightarrow}
\Ch_d(K)\stackrel{\partial_d}{\longrightarrow}\cdots
\stackrel{\partial_{n+1}}{\longrightarrow}\Ch_{n}(K)
\stackrel{\partial_{n}}{\longrightarrow}\Ch_{n-1}(K)
\to\cdots
\stackrel{\partial_{2}}{\longrightarrow}\Ch_{1}(K)
\stackrel{\partial_{1}}{\longrightarrow}\Ch_{0}(K)
\stackrel{\partial_{0}}{\longrightarrow}\Z
\longrightarrow 0
\longrightarrow \cdots
$$ 

The following is a standard result (see, for example, \cite[Lemma II.2.15]{FP:SimpTop}), and it says that 
the above sequence is a \emph{chain complex}.

\begin{prop}
For all $n\in\Z$, $\partial_{n}\circ\partial_{n+1}=0$.
\end{prop}

A consequence of this result is that the image of one boundary homomorphism is contained in the kernel of the next one, i.e.
$$
\im\, \partial_{n+1}\subset \ker \partial_{n}.
$$
Elements of $\im \, \partial_{n+1}$ are called \emph{$n$-boundaries} and those of $\ker \partial_{n}$ are called \emph{$n$-cyles}. The former should be thought of topologically as collections of simplices that form a boundary of another collection of simplices, and the latter as collections of simplices that have no boundary. Groups of $n$-boundaries and $n$-cycles are also denoted by $B_n(K)$ and $Z_n(K)$, respectively.

\begin{defin}\label{D:Homology}
For $n\geq 0$, the \emph{$n$th (reduced) homology group of $K$} is the quotient 
$$
\Ho_n(K)=\frac{Z_n(K)}{B_n(K)}=\frac{\ker \partial_{n}}{\im\, \partial_{n+1}}.
$$
\end{defin}

\begin{defin}\label{D:RealizationHomology}
For $K$ a simplicial complex, define the \emph{$n$th (reduced) homology group of the space $|K|$} to be the $n$th (reduced) homology group of $K$, i.e.
\begin{equation}\label{E:TwoHomologies}
\Ho_n(|K|)=\Ho_n(K).
\end{equation}
\end{defin}

If $\Ho_n(K)=0$ for all $n$, then $K$ (or $|K|$) is said to be \emph{acyclic}.

\begin{rems}\ 
\begin{itemize}
\item A more usual approach in algebraic topology is to first define the (simplicial) homology groups for any space $X$. In that case, equation \eqref{E:TwoHomologies} becomes a theorem stating that the notions of homology for spaces and abstract simplicial complexes coincide when the space happens to be a realization of a complex.
\item Our definitions give what is known as the \emph{reduced} homology. Namely, if a simplicial complex $K$ is connected, then the homology as we have defined it gives  $\Ho_0(K)=0$. If we set $\Ch_{-1}(K)=0$ and $\partial_0=0$, then we would get $\Ho_0(K)=\Z$, while all other homology groups of $K$ would be the same.  This version is the \emph{unreduced} homology. We have chosen the work with the reduced version for convenience; some statements involving the Betti numbers of a complex are easier to state this way.
\end{itemize} 
\end{rems}

\begin{example}\label{Ex:CircleHomology}
Suppose $K$ is a triangle formed by three edges with vertices $v_0$, $v_1$, and $v_2$. Then $\Ch_1(K)=\Z^3$, generated by the vertices $[v_i]$, and $\Ch_0(K)=\Z^3$, generated by $[v_0, v_1]$, $[v_1,v_2]$, $[v_0,v_2]$. Then
\begin{align*}
\Ho_0(K) & =\frac{Z_0(K)}{B_0(K)}=\frac{\langle [v_0], [v_1], [v_2] \rangle  }{\langle  [v_0]- [v_1], [v_0]- [v_2]\rangle}=0 \\
\Ho_1(K) & =\frac{Z_1(K)}{B_1(K)}=\frac{\langle [v_0, v_1]+[v_0, v_2]-[v_0, v_2] \rangle }{\langle 0\rangle}=\Z
\end{align*}
All other homology groups are trivial.

The geometric realization of this complex is topologically the circle $S^1$, so this example therefore computes the homology of $S^1$. By taking $K$ to be the boundary of the $d$-simplex, a similar calculation shows that the only nontrivial homology of the $d$-sphere $S^d$ is $\Ho_0(S^d)=\Ho_d(S^d)=\Z$.
\end{example}

The following useful results are standard. 
%(see \cite[Thorem II.4.9]{FP:SimpTop} and \cite[Theorem 5.16]{K:CombAlgTop})

\begin{prop}\label{P:HomologyOfOperations}
For $K$ and $L$ simplicial complexes and all $n$, 
\begin{enumerate}
\item $\Ho_n(K\vee L)=\Ho_n(K)\bigoplus \Ho_n(L)$
\item $\Ho_n(CK)=0$
\item $\Ho_{n+1}(\Sigma K)\cong \Ho_n(K)$
\end{enumerate}

\end{prop}

Because $\Ch_n(K)$ is a free group with a finite number of generators, we get that $\Ho_n(K)$ is a direct sum of some number of copies of $\Z$ (and possibly some finite cyclic groups).  The number of copies of $\Z$ is the rank of $\Ho_n(K)$ and is called the \emph{$n$th Betti number}, denoted by $\beta_n$.

Heuristically, $\beta_n$ is the number of $n$-dimensional ``holes'' in $K$, since each $\Z$ in  $\Ho_n(K)$ corresponds to an $n$-cycle which, on the level of $|K|$, can be thought of as an $n$-dimensional hole that is not a boundary, i.e.~it is not ``filled in''. The zeroth Betti number $\beta_0$ is simply one less than the number of connected components of $|K|$.

The alternating sum of the Betti numbers is called the \emph{Euler characteristic} of $K$.

If $f\colon K\to L$ is a simplicial map, we get for each $n\geq 0$ an induced map  
\begin{align*}
\Ch_n(f)\colon \Ch_n(K) & \longrightarrow \Ch_n(L) \\
\sum_{i\in I} c_i\sigma^i &\longmapsto \sum_{i\in I} c_i f(\sigma^i)
\end{align*}
It is not hard to show that this map descends to homology, so that we have a well-defined homomorphism
$$
\Ho_n(f)\colon \Ho_n(K) \longrightarrow \Ho_n(L).
$$
Another way to say this is that homology is a functor 
$$
\Ho_*\colon \operatorname{SCom}\longrightarrow \operatorname{GrAb}
$$
from the category of simplicial complexes to the category of graded abelian groups. We have

\begin{thm}[\cite{FP:SimpTop}, Theorem III.2.1]\label{P:ContiguityHomology}
If $f,g\colon K\to L$ are contiguous maps, then $\Ho_n(f)=\Ho_n(g)$ for all $n$.
\end{thm}

%(Say how this implies $|K|$ is contractible if also $\pi_1$ is trivial, but not necessarily strongly contractible.

%Homology of a complex with a bunch of components is the direct sum of homologies.

If $K$ and $L$ are related by a sequence of elementary strong collapses and expansions, then $\Ho_n(K)=\Ho_n(L)$ for all $n$.  This follows since $K$ and $L$ have the same strong homotopy type and so $|K|$ and $|L|$ are (strongly) homotopy equivalent and hence have isomorphic homology.  This will be useful in \refS{Delegations}, where we will observe that certain simplicial complexes do not have the same homology and hence cannot be strongly equivalent.

In particular, if $L$ is strongly contractible, then $\Ho_n(L)=0$ for all $n$. So if $L=CK$ for some $K$, by \refC{ConeSuspensionStrong}, we recover part (2) of \refP{HomologyOfOperations}.

Since chains are vector spaces and boundary homomorphisms are linear maps determined by their action on vertices, homology can be computed using linear algebra.  Briefly, one considers the 
\emph{Laplacians of $K$}, 
 $\mathbf L_n = \mathbf A_n\mathbf A_n^T + \mathbf A_{n+1}\mathbf A_{n+1}^T$, where $\mathbf A_i$ are the adjacency matrices of $K$ (these capture how simplices fit together in $K$). The kernel of the $n$th Laplacian then gives $\beta_n$. There are ready-made online tools for computing the homology of simplicial complexes using this method.\footnote{See, for example, \href{https://jeremykun.com/2013/04/10/computing-homology/}{https://jeremykun.com/2013/04/10/computing-homology/}.}

\bibliographystyle{amsplain}

%\bibliography{/Users/ismar/Dropbox/Bibliography}

%
%\def\cprime{$'$} \def\cprime{$'$}
%\providecommand{\bysame}{\leavevmode\hbox to3em{\hrulefill}\thinspace}
%\providecommand{\MR}{\relax\ifhmode\unskip\space\fi MR }
%% \MRhref is called by the amsart/book/proc definition of \MR.
%\providecommand{\MRhref}[2]{%
%  \href{http://www.ams.org/mathscinet-getitem?mr=#1}{#2}
%}
%\providecommand{\href}[2]{#2}

\end{document}